\documentclass[11pt]{article}
\usepackage{fullpage}
\usepackage{amsmath,amsthm,amsfonts,amssymb}
\usepackage{graphicx,enumerate}
\usepackage[title]{appendix}
\usepackage{hyperref}
\hypersetup{CJKbookmarks=true,
bookmarksnumbered=true,
bookmarksopen=true,
colorlinks=true,
citecolor=red,
linkcolor=blue,
anchorcolor=red,
urlcolor=blue}

\usepackage{caption}
\usepackage{booktabs}
\usepackage{subcaption}
\usepackage{enumitem}
\usepackage{natbib}

\makeatletter
\DeclareFontEncoding{LS1}{}{}
\DeclareFontSubstitution{LS1}{stix}{m}{n}
\DeclareMathAlphabet{\mathscr}{LS1}{stixscr}{m}{n}
\makeatother

\newcommand{\expect}{\mathbb{E}}

\newcommand{\prob}{\mathbb{P}}

\usepackage{tikz}
\usepackage{tikz-qtree}
\usepackage{mathtools}
\usetikzlibrary{calc}
\usetikzlibrary{decorations.pathreplacing,decorations.markings}
\tikzstyle{dot}=[circle,fill,black,inner sep=1pt]

\newtheorem{theorem}{Theorem}
\newtheorem{lemma}{Lemma}
\newtheorem{proposition}{Proposition}

\newtheorem{definition}{Definition}

\usepackage{color}
\usepackage{algorithm}% http://ctan.org/pkg/algorithms
\usepackage{algorithmic}% http://ctan.org/pkg/algorithms
\usepackage{xspace}

\renewcommand{\tilde}{\widetilde}

\renewcommand{\root}{\textrm{\o}}

\title{The Planted Spanning Tree Problem}

\author{Mehrdad Moharrami, Cristopher Moore, and Jiaming Xu\thanks{
		Accepted for presentation at the Conference on Learning Theory (COLT) 2025. 
            M.\ Moharrami is with Computer Science Department, University of Iowa, Iowa City, IA, USA, \texttt{moharami@uiowa.edu}.
		C.\ Moore is with Santa Fe Institute, Santa Fe, NM, USA, \texttt{moore@santafe.edu}.
		J.\ Xu is with The Fuqua School of Business, Duke University, Durham, NC, USA, \texttt{jx77@duke.edu}.
}}
\date{\today}

\begin{document}
	
\pgfdeclarelayer{background}
\pgfdeclarelayer{foreground}
\pgfsetlayers{background,main,foreground}

\maketitle

\begin{abstract}
We study the problem of detecting and recovering a planted spanning tree $M_n^*$ hidden within a complete, randomly weighted graph $G_n$.  Specifically, each edge $e$ has a non-negative weight drawn independently from $P_n$ if $e \in M_n^*$ and from  
$Q_n$ otherwise, where $P_n \equiv P$ is fixed and $Q_n$ scales with $n$ such that its density at the origin satisfies $\lim_{n\to\infty} n Q'_n(0)=1.$  
We consider two representative cases: when $M_n^*$ is either a uniform spanning tree or a uniform Hamiltonian path. 
We analyze the recovery performance of the minimum spanning tree (MST) algorithm and derive a fixed-point equation that characterizes the asymptotic fraction of edges in $M_n^*$ successfully recovered by the MST as  $n \to \infty.$ Furthermore, we establish the asymptotic mean  weight of the MST, extending Frieze's $\zeta(3)$ result to the planted model. {Leveraging this result, we design an efficient test based on the MST weight and show that it can distinguish the planted model from the unplanted model with vanishing testing error as  $n \to \infty.$} Our analysis relies on an asymptotic characterization of the local structure of the planted model, employing the framework of local weak convergence.
\end{abstract}


\section{Introduction}\label{sec:intro}

The study of planted problems in random graph models has a rich and extensive history, focusing on uncovering hidden structures planted within an otherwise random graph. Notable examples include the planted clique problem~\citep{alon1998finding} and community detection in stochastic block models~(e.g., \citet{DecelleKrzakalaMooreZdeborova2011Asymptotic}). A central quantity of interest in these problems is the asymptotic overlap between the planted object and the estimated one in the large-graph limit. 
While significant progress has been made, much of the work has concentrated on characterizing recovery thresholds --  critical points at which the asymptotic overlap transits from $0$ to positive values or reaches $1$~\citep{mossel2015reconstruction,massoulie2014community,mossel2018proof,DingWuXuYang2023Planted,mossel2023sharp}.
However, in many of these problems characterizing the exact value of the asymptotic overlap remains a formidable mathematical challenge.

Recent years have witnessed significant breakthroughs in this direction, with several novel techniques being developed. For instance, the interpolation method has been successfully applied to 
stochastic block models with disassortative communities, i.e., where inter-community edges are more prevalent than intra-community ones~\citep{coja2017information}. Nevertheless, this method is often complex and challenging to implement -- for instance, its applicability to stochastic block models with assortative communities remains an open question. Another closely related approach involves the analysis of belief propagation algorithms, which has recently been used to establish the asymptotic overlap in the two-group stochastic block model~\citep{yu2023ising}. 
This technique typically relies on recursive distributional equations involving infinite-dimensional probability measures, making the exact analysis highly challenging.
A third and promising technique is based on local weak convergence~\citep{AldousSteel2001Objective}. This framework was employed in \cite{MoharramiMoorXu2021Planted} to rigorously derive the asymptotic overlap for the problem of recovering a planted matching within a complete, randomly weighted bipartite graph. By leveraging the machinery of local weak convergence, they characterized the asymptotic overlap between the planted matching and the minimum-weight matching (which corresponds to the maximum-likelihood estimator) through the solution of a system of ordinary differential equations.

In this paper, we extend the framework of local weak convergence to study the problem of recovering a planted spanning tree within a complete, randomly weighted graph. 
\begin{definition}[Planted spanning tree model]~\\
{\bf Given:} $n \ge 1$, and  two distributions $P_n,Q_n$ supported on the non-negative real line. \\
{\bf Observation:} A randomly weighted, undirected complete graph $G_n=([n], W)$ with a planted spanning tree $M_n^*$ such that the weight of each edge $W_e$ is independently distributed as $P_n$ if  $e \in M^*$ and as $Q_n$ otherwise. \\
{\bf Goal}: Recover the planted spanning tree $M_n^*$ from the observed graph $G_n$. 
\end{definition}
We assume $P_n \equiv P$ is a fixed continuous distribution, while $Q_n$ scales with $n$ such that its density at the origin satisfies $\lim_{n\to \infty} n Q'_n(0)=1$.\footnote{Assuming the limit to be 
$1$ is solely for simplicity of presentation. As we will see in Section~\ref{sec:proof}, the condition $\lim_{n \to \infty} n Q_n'(0) = 1$ ensures that for any given vertex, the weights of its incident unplanted edges, when sorted in increasing order, converge to the arrival times of a rate-1 Poisson process as $n \to \infty$. This is the only property of $Q_n$ needed to characterize the local weak limit and hence the asymptotic overlaps.  
The results and analysis readily extend to the more general case where  $\lim_{n\to \infty} n Q'_n(0)=C$ for some constant $C>0.$} This scaling ensures that the minimum magnitude of unplanted edge weights incident to a given vertex is typically of the same order as the planted edge weights. A canonical example, referred to as the exponential model, assumes $P$ and $Q_n$ are  exponential distributions with mean $\mu$ and $n$, respectively. For generating the planted spanning tree $M_n^*$, we consider two representative scenarios: 
\begin{itemize}
    \item {\bf Uniform spanning tree}: $M_n^*$ is chosen uniformly  from all $n^{n-2}$ spanning trees;
    \item {\bf Uniform Hamiltonian path}:
    $M_n^*$ is  chosen uniformly from all $n!/2$ Hamiltonian paths. 
\end{itemize}
In the uniform spanning tree model, the planted spanning tree may belong to different isomorphism classes, whereas the uniform Hamiltonian path model is a special case restricted to the isomorphism class of $n$-paths.

We estimate $M_n^*$ using the minimum-weight spanning tree (MST), denoted by $M_n$. This corresponds to the maximum-likelihood estimator when $M_n^*$ is the uniform spanning tree and the edge weights follow exponential distributions. The MST estimator can be computed in linear time using greedy algorithms, making it an appealing choice even under the uniform Hamiltonian path model, as finding the Hamiltonian path of minimum weight -- namely, the Traveling Salesman Problem -- is NP-hard in the worst case.

Define the \emph{overlap} between $M_n$ and $M_n^*$ as the fraction of the common edges:
\[
\text{overlap}(M_n, M^*_n) = \frac{1}{n-1} |M_n \cap M_n^*|.
\]
We obtain the exact expression of $\lim_{n \to \infty}\mathbb{E}[\text{overlap}(M_n, M^*_n)]$ in terms of the unique solution to certain fixed-point equations. 
 Our analysis builds upon the local weak convergence framework developed in \cite{AldousSteel2001Objective,Steele2002Minimal} and follows closely the methodology of  \cite{MoharramiMoorXu2021Planted} for the planted matching model. A key distinction, however, is that while the local structure of the matching is trivial, the local structure of the uniform spanning tree converges to a non-trivial asymptotic object, known as the \emph{skeleton tree}\footnote{A skeleton tree consists of a unique infinite path extending from the root to infinity, where each vertex along this path is attached with an independent Poisson Galton-Watson tree with mean $1$.}~\cite{Aldous1991Asymptotic,Grimmett1980Random,AldousSteel2001Objective}. Furthermore, unlike the planted matching model, where the asymptotic overlap exhibits a phase transition at a critical threshold, we find that no such phase transition occurs in our setting.

Furthermore, we derive the asymptotic value of the mean weight of the minimum weight spanning tree, $\lim_{n\to \infty}  \mathbb{E}[w(M_n)]$, 
where $w(M_n)$ is defined as 
$$
w(M_n)= \frac{1}{n-1} \sum_{e \in M_n} W_e. 
$$
This result generalizes the classical finding of~\cite{Frieze1985Value} for the unplanted model with $P_n=Q_n$, which established that $\lim_{n\to \infty}  \mathbb{E}[w(M_n)] =\zeta(3)$,  where $\zeta(s) = \sum_{i=1}^{\infty} i^{-s}$ is the Riemann zeta function. Leveraging this result, we propose an efficient test based on $w(M_n)$ and show that it can distinguish the planted model from the unplanted model with vanishing testing error as $n \to \infty.$

Our approach is broadly applicable, as it can be extended to characterize the asymptotic overlap of the MST estimator in recovering other planted structures, provided that these structures exhibit well-behaved local weak convergence. More generally, we believe this work opens new avenues for applying local weak convergence to characterize asymptotic overlaps in a wide range of planted problems in random graphs.

Finally, we note that \cite{MassouliStephanTowsley2019Planting} studied a closely related problem of detecting and recovering planted trees in Erd\H{o}s-R\'enyi random graphs. However, the planted trees in their setting are much smaller in size and not spanning, and their analysis primarily focused on detection and recovery thresholds.

\section{Asymptotic Overlap and Mean Weight of MST} The asymptotic overlap between the minimum spanning tree $M_n$ and the planted tree $M^*_n$ in $G_n$ is characterized by fixed-point equations, with the form of the equations varying across different scenarios.
\begin{theorem}\label{thm:main}
    Let $P $ denote the edge weight distribution of $M_n^*$ with the cumulative distribution function $F$ and valid probability density function. The asymptotic overlap of the minimum-weight spanning tree $M_n$ is given by
    \begin{align}
        \lim_{n \to \infty} \mathbb{E}[\text{\rm overlap}(M_n, M_n^*)] = \int_0^\infty \big(1 - (1 - p_{-}(s))(1 - p_{+}(s))\big) \, dF(s),
    \end{align}
    where $p_{-}(\cdot)$ and $p_{+}(\cdot)$ are characterized by fixed-point equations, depending on the underlying structure of $M_n^*$:
    \begin{enumerate}[label = (\arabic*)]
        \item \textbf{Uniform spanning tree}: In this case, $p_{-}(\cdot) \eqqcolon p_U(\cdot)$ and $p_{+}(\cdot) \eqqcolon p_B(\cdot)$ are 
        the smallest fixed point to the following equations:
        \begin{align}
        \begin{aligned}
            p_B(s) &= \exp\big(-s(1 - p_U(s))\big) \exp\big(-(1 - p_B(s))\big),\\
            p_U(s) &= \exp\big(-s(1 - p_U(s))\big) \exp\big(-(1 - p_B(s))\big) \big(1 - F(s) + F(s)p_U(s)\big).
        \end{aligned}\label{eq:fixed_point_US}
        \end{align}
        \item \textbf{Uniform Hamiltonian path}: In this case, $p_{-}(\cdot) = p_{+}(\cdot) \eqqcolon p(\cdot)$ is the
        smallest fixed point of the following equations:
        \begin{align}
        \begin{aligned}
            p(s) &= \exp\big(-s(1 - q(s))\big) \big(1 - F(s) + F(s)p(s)\big),\\
            q(s) &= \exp\big(-s(1 - q(s))\big) \big(1 - F(s) + F(s)p(s)\big)^2. 
        \end{aligned}\label{eq:fixed_point_HP}
        \end{align}
    \end{enumerate}
    In both cases, the smallest fixed points of the equations can be found iteratively, starting from the all-zero function.
\end{theorem}

As we will see in the proof, $p_+(\cdot)$ and $p_{-}(\cdot)$ are extinction probabilities of certain branching processes. Notice that the all‐ones function is always a solution to the fixed‐point equations appearing in Theorem~\ref{thm:main}. A natural question is then: How many solutions can there be?

For a Galton–Watson branching process in the supercritical regime, it is well known that the generating function of the offspring distribution has at most two fixed points: the trivial fixed point at $1$, and a nontrivial fixed point in $(0,1)$ that corresponds to the extinction probability of the process. The following proposition shows an analogous phenomenon for the fixed‐point equations in Theorem~\ref{thm:main}: these equations can have at most two solutions, one of which is the all‐ones function. The details of the derivation are presented in Appendix~\ref{app:fixedpoint}.

\begin{proposition}\label{prop:one_nontrivial_solution}
Depending on the value of $s$, the fixed-point equations \eqref{eq:fixed_point_US} and \eqref{eq:fixed_point_HP} admit \emph{at most one} nontrivial solution in $(0,1)$:
\begin{enumerate}[label = (\arabic*)]
    \item \textbf{Uniform spanning tree:} The nontrivial solution is uniquely given by
    \begin{align*}
        p_B(s) = x^*(s),\text{ and }
        p_U(s) = \frac{x^*(s)\,\bigl(1 - F(s)\bigr)}{1 \;-\; x^*(s)\,F(s)},
    \end{align*}
    where $x^*(0)=1$ and for $s>0$, 
    $x^*(s)$ is the unique solution in $(0,1)$ to the fixed-point equation $\phi_s(x) = x$, with
    \begin{align*}
        \phi_s(x) = \exp \Bigg(-(1 - x)\Bigg(\frac{s}{1 - F(s)\,x} + 1\Bigg)\Bigg).
    \end{align*}
    
    \item \textbf{Uniform Hamiltonian path:} If $s > \tfrac{1-F(s)}{1+F(s)}$, then the nontrivial solution is uniquely given by
        \begin{align*}
          p(s) = x^*(s),\text{ and }
          q(s) = x^*(s)\,\bigl(1 - F(s) + F(s)\,p(s)\bigr),
        \end{align*}
        where $x^*(s)$ is the unique solution in $(0,1)$ to the fixed-point equation $\phi_s(x) = x$, with
        \begin{align*}
          \phi_s(x) = \exp \Big(-s\big(1 - x (1 - F(s) + F(s) x )\big)\Big)
          \,\times\, (1 - F(s) + F(s) \, x).
        \end{align*}
        If $s \le \tfrac{1-F(s)}{1+F(s)}$, then $p(s) = q(s) = 1$ is the only solution.
\end{enumerate}
\end{proposition}

For the unplanted model with $P_n \equiv Q_n$, \cite{Frieze1985Value} showed that the mean weight of the minimum spanning tree, $ \mathbb{E}[w(M_n)]$, converges to $\zeta(3)$ as $n \to \infty$.
%, where $\zeta(s) = \sum_{i=1}^{\infty} i^{-s}$ is the Riemann zeta function. 
We extend this classical result to the planted model, characterizing $\lim_{n\to\infty}  \mathbb{E}[w(M_n)]$.
\begin{theorem}\label{thm:asymcost}
    Suppose the continuous distribution $P$ of the planted edges has a finite mean.
    \begin{enumerate}[label = (\arabic*)]
        \item \textbf{Uniform spanning tree}: Given $p_U(\cdot)$ and $p_B(\cdot)$ as defined in Theorem~\ref{thm:main}, we have
        \begin{flalign*}
            \lim_{n\to\infty}
             \mathbb{E}[w(M_n)]
            %\mathbb{E} \left[ \frac{1}{n-1} \sum_{e} W_e \, \mathbf{1}_{\{e \in M_n\}} \right] 
            =\int_0^\infty s \big(1 - (1 - p_{U}(s))(1 - p_{B}(s))\big) \, dF(s) + \frac{1}{2}\int_0^\infty s \big(1 - (1 - p_{U}(s))^2\big) \, ds.
        \end{flalign*}
        \item \textbf{Uniform Hamiltonian path}: Given $p(\cdot)$ and $q(\cdot)$ as defined in Theorem~\ref{thm:main}, we have
        \begin{flalign*}
            \lim_{n\to\infty}
             \mathbb{E}[w(M_n)]
            %\mathbb{E} \left[ \frac{1}{n-1} \sum_{e} W_e \, \mathbf{1}_{\{e \in M_n\}} \right] 
           = \int_0^\infty s \big(1 - (1 - p(s))^2\big) \, dF(s) + \frac{1}{2}\int_0^\infty s \big(1 - (1 - q(s))^2\big) \, ds.
        \end{flalign*}
    \end{enumerate}
\end{theorem}

Figure~\ref{fig:planted-simuvsth} compares the theoretical results with numerical simulations, where the planted distribution $P$ follows an exponential distribution with mean $\mu$. Observe that the algorithm's performance is remarkably close across the two cases, though not identical. Closer inspection reveals that the overlap between the uniform spanning tree and the MST is greater for small values of $\mu$, whereas the overlap between the uniform Hamiltonian path and the MST is greater for large values of $\mu$. 
Furthermore, as $\mu \to 0$, the overlap converges to $1$, while the mean weight of the MST converges to $0$. Conversely, as $\mu \to \infty$, the overlap vanishes, and the mean weight of the MST approaches $\zeta(3) \approx 1.202$. See Table~\ref{table:datapoints} in Appendix~\ref{app:table} for exact data points corresponding to Figure~\ref{fig:planted-simuvsth}.

Another observation is the absence of a phase transition in the planted tree problem when using the MST algorithm. As shown in Figure~\ref{fig:planted-simuvsth}, the overlap converges to $1$ as $\mu \to 0$. However, for any $\mu > 0$, we have $\lim_{n \to \infty} \mathbb{E}[\text{overlap}(M_n, M_n^*)] < 1$. In contrast, the planted matching problem~\cite{MoharramiMoorXu2021Planted,DingWuXuYang2023Planted} exhibits a phase transition at the  critical value  of $1/4$: below this threshold, almost perfect recovery is achievable using the minimum matching algorithm; whereas at or above this threshold, almost perfect recovery becomes information-theoretically impossible for any algorithm.

\begin{figure}[h!]
    \centering
    \begin{subfigure}[b]{0.48\textwidth}
        \centering
        \includegraphics[width=\textwidth]{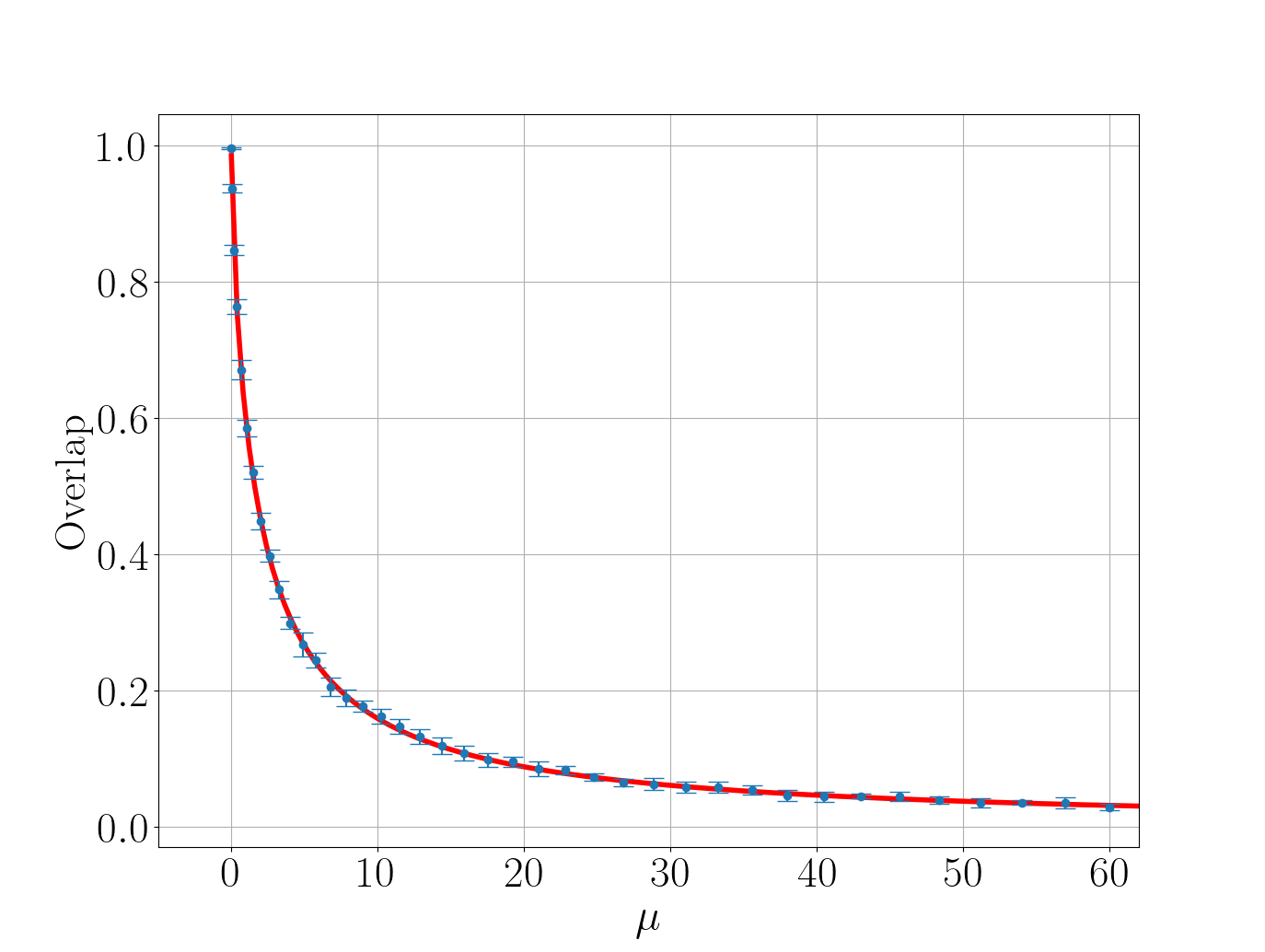}
        \caption{Overlap for uniform Hamiltonian path}
    \end{subfigure}
    \hfill
    \begin{subfigure}[b]{0.48\textwidth}
        \centering
        \includegraphics[width=\textwidth]{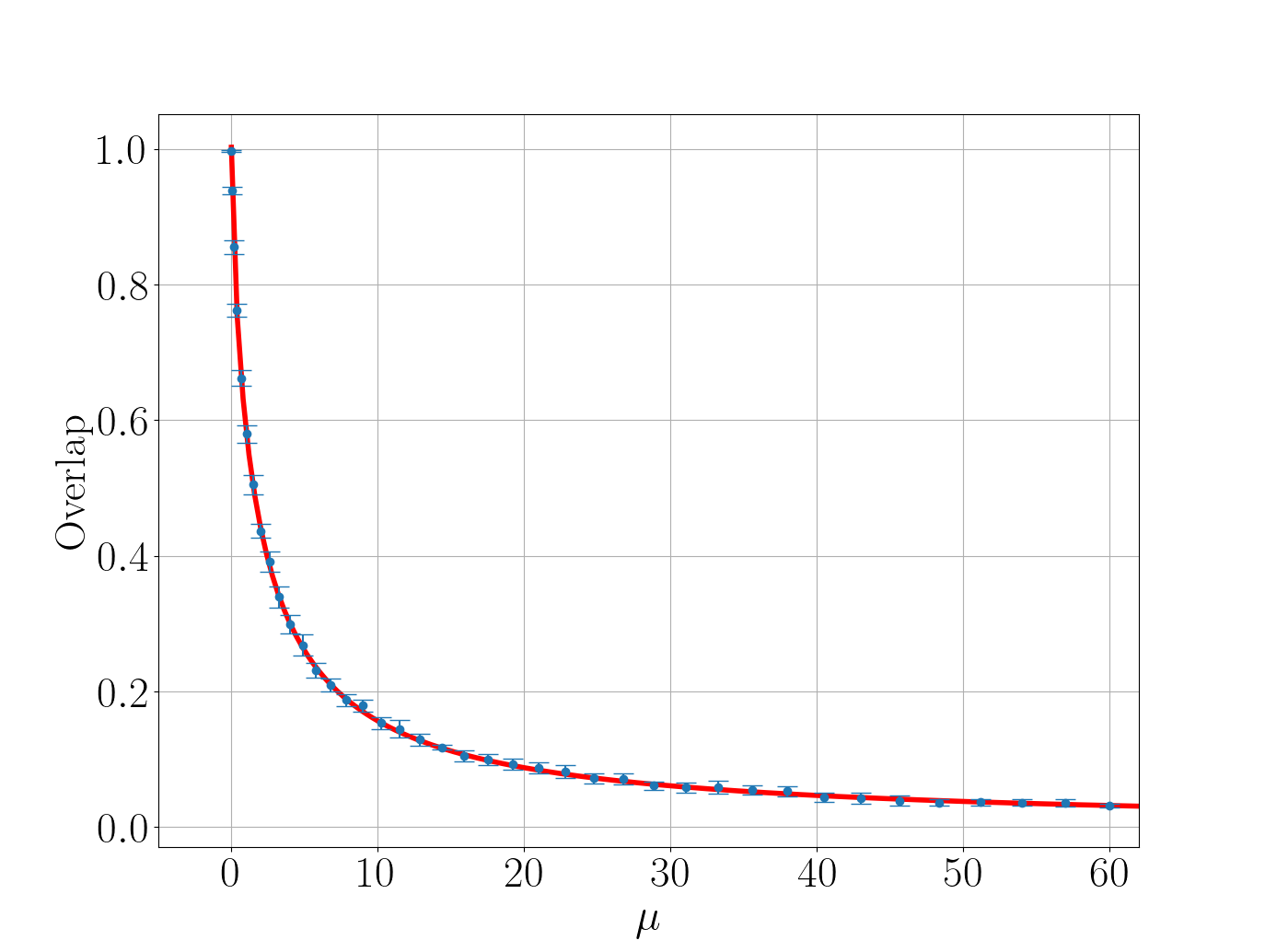}
        \caption{Overlap for uniform spanning tree}
    \end{subfigure}\\
    \begin{subfigure}[b]{0.48\textwidth}
        \centering
        \includegraphics[width=\textwidth]{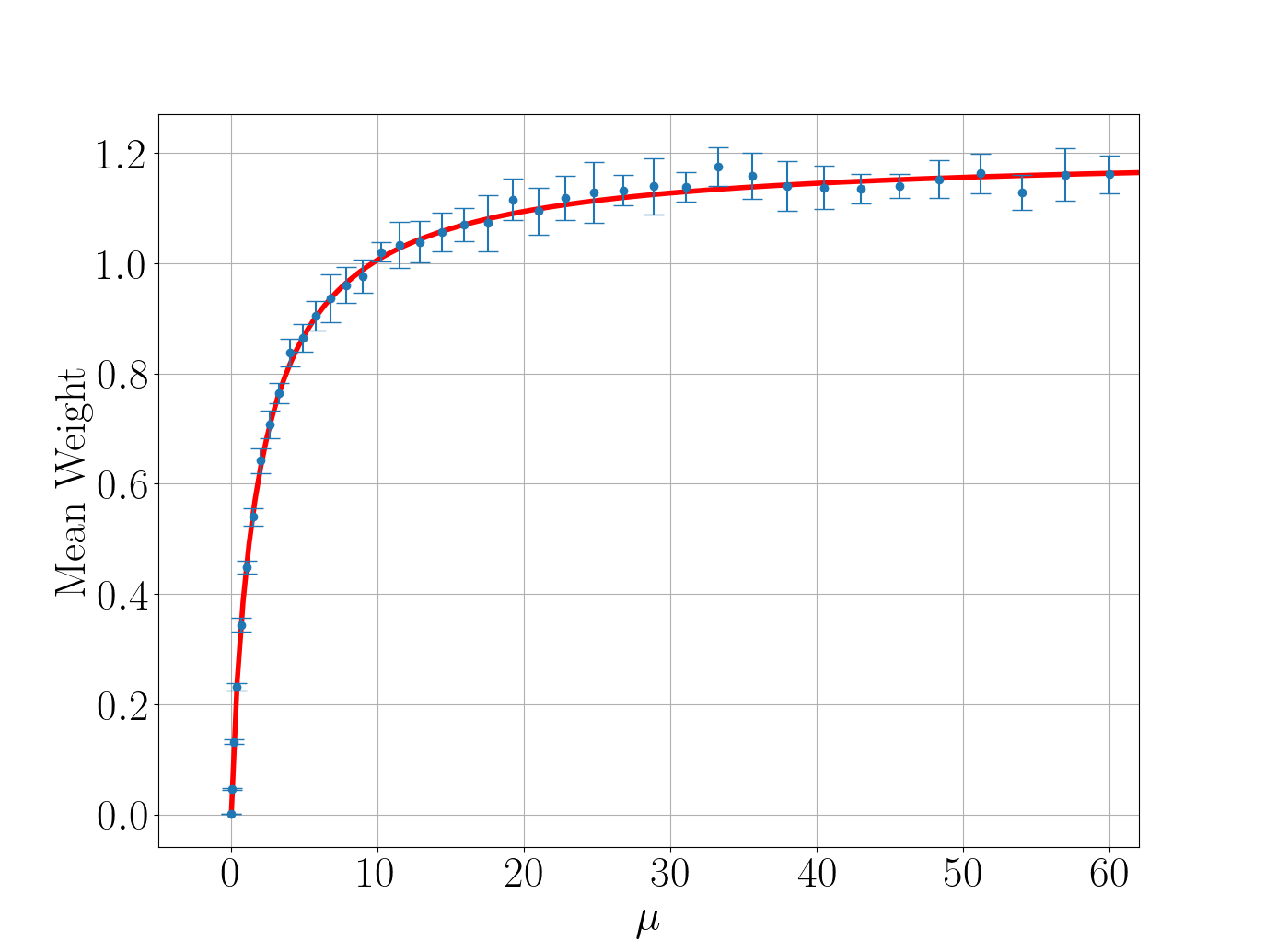}
        \caption{$w(M_n)$ for uniform Hamiltonian path}
    \end{subfigure}
    \hfill
    \begin{subfigure}[b]{0.48\textwidth}
        \centering
        \includegraphics[width=\textwidth]{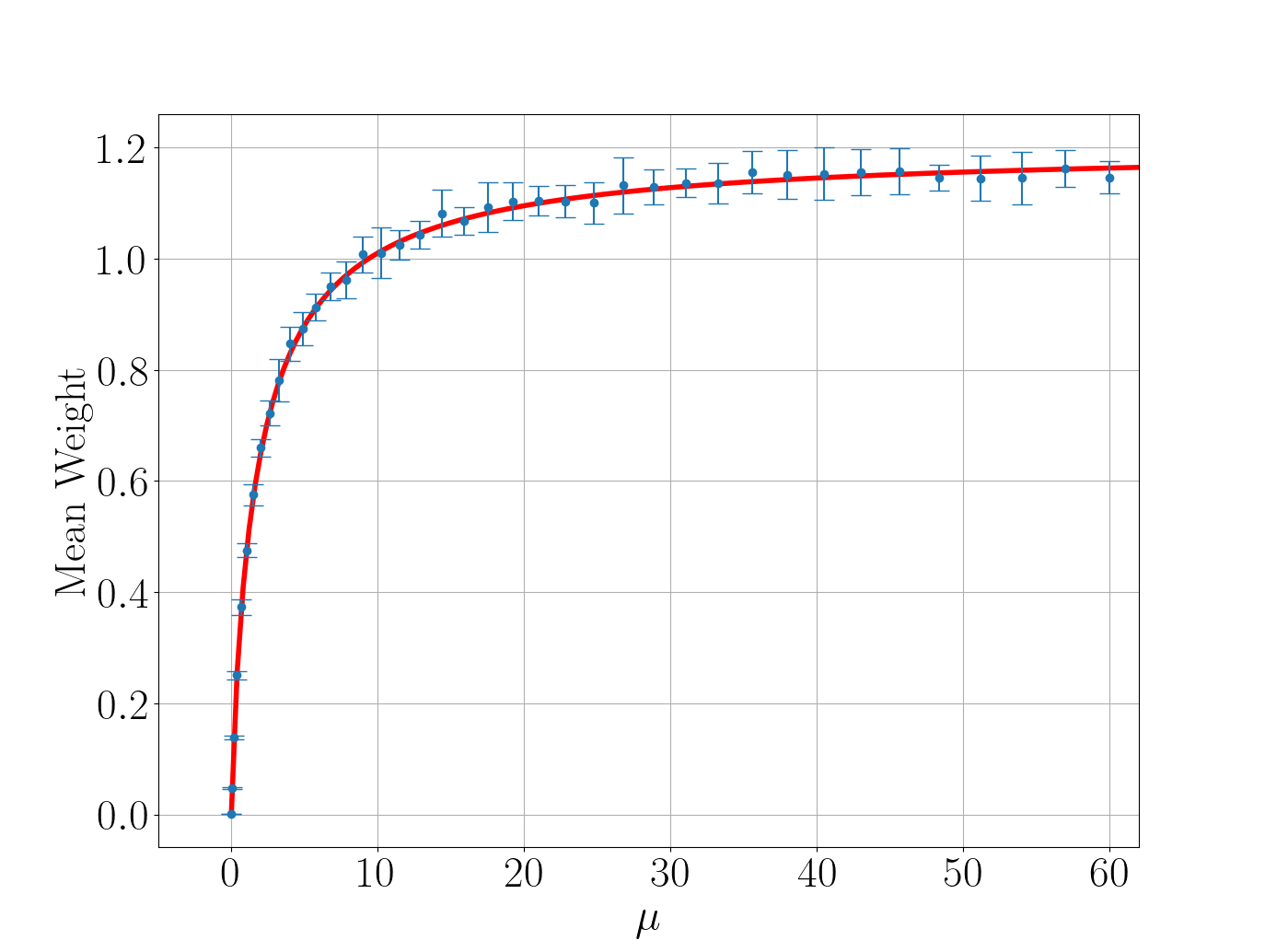}
        \caption{$ w(M_n)$ for uniform spanning tree}
    \end{subfigure}
    \caption{The solid red line represents the theoretical values obtained by numerically solving the fixed-point equations and then taking numerical integration. The blue dots indicate the empirical means of the corresponding quantity, computed from complete graphs generated from the planted exponential model. The error bars depict the standard deviation of the distribution, highlighting that the empirical quantities are tightly concentrated around their asymptotic expected values.}
    \label{fig:planted-simuvsth}
\end{figure}

Leveraging Theorem~\ref{thm:asymcost}, we design a test based on $w(M_n)$ to distinguish between the planted model and the unplanted one. Specifically, consider the following null and alternative hypotheses: 
\begin{align}
\mathcal{H}_0: P_n \equiv Q_n,   \qquad \mathcal{H}_1: P_n \equiv P.
\end{align}
Given that $G_n$ is generated under either $\mathcal{H}_0$ or $\mathcal{H}_1$, our goal is to determine the underlying hypothesis based on the observation of $G_n$. 

For a given positive constant $\epsilon>0$, define a test $\phi: G_n \to \{0,1\}$ that outputs $0$ if $w(M_n) \ge \zeta(3)-\epsilon$ and $1$ otherwise. The following theorem establishes that the sum of  Type-I and Type-II errors of this test vanishes as $n \to \infty.$

\begin{theorem}[Hypothesis Testing]\label{thm:hypothesis_testing}
Assume $\lim_{n\to\infty} \mathbb{E}_{\mathcal{H}_1}[w(M_n)] \le \zeta(3) -  2\epsilon$ for some positive constant $\epsilon>0.$ 
%$\omega(\log(n)/\sqrt{n}) \le \epsilon_n \le o(1)$. 
Then 
$$
\lim_{n\to \infty} \left(\mathbb{P}_{\mathcal{H}_0}\left\{ 
\phi(G_n)=1
\right\}
+ \mathbb{P}_{\mathcal{H}_1}\left\{ 
\phi(G_n)=0
\right\} \right) =0.
$$
\end{theorem}
%We prove Theorem~\ref{thm:hypothesis_testing} 
The proof follows by 
applying Talagrand's concentration inequality with a proper truncation on the edge weights. See Appendix~\ref{sec:proof_hypothesis} for details. 
%first establishing a variance bound  $\text{Var}[w(M_n)]=O(\log^2/n)$ under both the null and alternative hypothesis,
%and then applying Chebyshev's inequality together with Theorem~\ref{thm:asymcost}.

\section{Proof of the Asymptotic Overlap and Mean Weight} \label{sec:proof}
    The proof leverages the local weak convergence framework developed in \cite{AldousSteel2001Objective,Steele2002Minimal} for analyzing the average edge weight of the minimum spanning tree in a randomly weighted complete graph $K_n$, where the edge weights are drawn from an exponential distribution with mean $n$. Specifically, this approach establishes the asymptotic value as $\zeta(3)$,  providing a novel perspective on a classical result~\cite{Frieze1985Value}. 

    In this section, we present the main idea behind the proof and refer the reader to the appendix for a rigorous treatment of the underlying concepts. The proof of the theorem proceeds in three steps:
    \begin{enumerate}
        \item \textbf{Convergence of the Planted Graph:} We first show that the planted graph $G_n$ converges in the local weak sense to an infinite tree $T_\infty$. This builds on similar results established for the planted matching problem~\cite{MoharramiMoorXu2021Planted}. As $n \to \infty$, the complete randomly weighted graph $K_n$
        converges to the Poisson Weighted Infinite Tree (PWIT)~\cite{AldousSteel2001Objective}; a uniform Hamiltonian path in $K_n$ converges to an infinite line, and a uniform spanning tree in $K_n$ converges to the skeleton tree~\cite{Aldous1991Asymptotic,Grimmett1980Random,AldousSteel2001Objective}. We rely on these results to establish the convergence of the planted model.
        \item \textbf{Convergence of the Minimum Spanning Tree:} Next, we show that the minimum spanning tree $M_n$ converges in the local weak sense to the \textit{minimum spanning forest} of $T_\infty$. This step relies on the results of~\cite{Steele2002Minimal}, which shows that the local weak convergence of a sequence of finite graphs to an asymptotic object implies the local weak convergence of the minimum spanning trees of the finite graphs to the minimum spanning forest of the asymptotic object.
  We sketch the core idea of defining the minimum spanning forest on $T_\infty$ and demonstrate its connection to the minimum spanning tree of the planted graph $G_n$. 
        \item \textbf{Detection Probability and Mean Weight:} Finally, we characterize the probability that a planted edge belongs to the minimum spanning forest of $T_\infty$ using fixed-point equations. This characterization enables us to derive the asymptotic overlap between the minimum spanning tree $M_n$ and the planted structure $M_n^*$. We then extend our analysis to calculate the asymptotic mean weight of $M_n$.
    \end{enumerate} 
    
    We adopt the notation from~\cite{MoharramiMoorXu2021Planted}, where $(\ell^\square_n, G^\square_n)$ denotes a planted model with $n$ vertices. More specifically, $G^\square_n = (V^\square_n, E^\square_n)$ is a finite graph with vertex set $V^\square_n$ and edge set $E^\square_n$, and $\ell^\square_n \colon E^\square_n \to \mathbb{R}_{\ge 0}$ is a weight function assigning nonnegative weights to each edge. The corresponding infinite-limit tree is denoted by $(\ell^\square_\infty, T^\square_\infty, \root)$, where $\root\in V^\square_\infty$ is the designated root vertex. Here, $\square \in \{S, H\}$ indicates the type of planted model: ``S'' for the uniform spanning tree model and ``H'' for the uniform Hamiltonian path model. Specifically, we denote the planted Hamiltonian path in $G^H_n$ by $M_n^{H,*}$ and the planted spanning tree in $G^S_n$ by $M_n^{S,*}$.

    Next, we introduce the notion of local weak convergence used in our analysis. This convergence captures the local structure around a randomly selected vertex, where ``local'' refers to the neighborhood of any fixed radius, measured by the total sum of edge weights (interpreted as distance). A precise formulation can be found in Appendix~\ref{app:localweakframework}.
    \begin{definition}[Local Weak Convergence (Informal)]
        A sequence of random finite weighted graphs $(\ell^\square_n, G^\square_n)$ with $n$ vertices converges locally weakly to a random infinite rooted tree $(\ell^\square_\infty, T^\square_\infty, \root)$ if, upon choosing a vertex uniformly at random in $(\ell^\square_n, G^\square_n)$, the neighborhood of that vertex (up to any fixed radius) converges in distribution to the neighborhood of the distinguished root $\root$ in $(\ell^\square_\infty, T^\square_\infty, \root)$.
    \end{definition}
    
    \subsubsection*{Step 1: Convergence of the Planted Graph.}
    We begin by defining the random infinite rooted trees $(\ell^H_\infty, T^H_\infty, \root)$ and $(\ell^S_\infty, T^S_\infty, \root)$. To determine the structure of the infinite tree, we choose a vertex uniformly at random in $(\ell^\square_n, G^\square_n)$  as the root, and analyze the convergence of its local neighborhood, defined in terms of $\ell^\square_n$. 

    First, let us define the infinite rooted tree $(\ell^H_\infty, T^H_\infty, \root)$ illustrated in Figure~\ref{fig:plantedH}. The key insight is that each vertex should have exactly two \emph{planted} neighbors, corresponding to the Hamiltonian path in the finite model. Specifically, the root vertex $\root$ is connected to two planted children by bold red edges and to an infinite sequence of \emph{unplanted} children by solid blue edges. We label the planted children of the root by $\{\tilde{1}, \tilde{2}\}$, and the unplanted children by $\{1, 2, 3, \dots\}$. For a planted child $\tilde{i} \in \{\tilde{1}, \tilde{2}\}$, the edge weight $\ell^H_\infty(\root,\tilde{i})$ is sampled independently from $P$, whereas for an unplanted child $i \in \{1,2,\ldots\}$, the edge weight $\ell^H_\infty(\root,i)$ is set to be the $i$th arrival time of a rate-$1$ Poisson process (The choice of rate-$1$ stems from the assumption that $\lim_{n\to \infty} nQ_n'(0)=1$). This construction mirrors the descendant structure of the root vertex in the PWIT~\citep{AldousSteel2001Objective}. 

    For non-root vertices, the descendant distribution depends on whether a vertex is planted or unplanted. An unplanted vertex has the same distribution for its descendants as the root; namely, it has two planted children (with edge weights from $P$) and infinitely many unplanted children (with edge weights from a Poisson process). A planted vertex, on the other hand, has exactly one planted child (with an edge weight drawn from $P$) and an infinite sequence of unplanted children (with edge weights determined by the arrivals of a rate-$1$ Poisson process). Figure~\ref{fig:plantedH} depicts this structure and the labeling scheme: planted children inherit their parent’s label with $\tilde{1}$ or $\tilde{2}$ appended, while unplanted children are labeled by the order in which they appear in the Poisson process.

    \begin{figure*}[h!]
    \centering
    \includegraphics[width=0.8\textwidth]{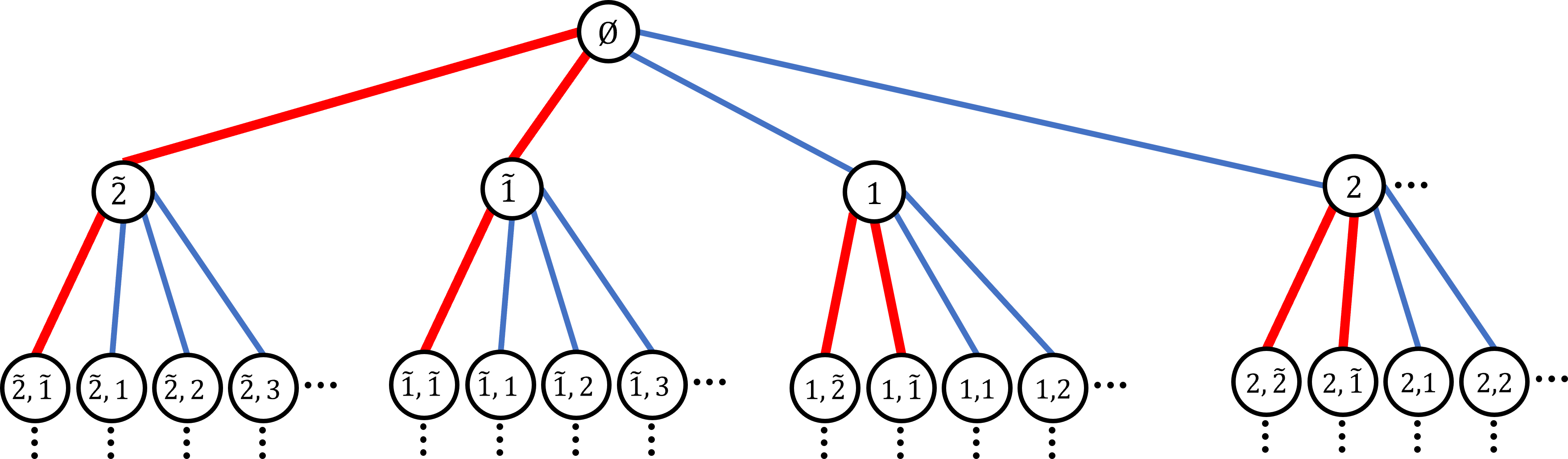}
    \caption{Structure of $(\ell^H_\infty, T^H_\infty)$. Planted edges are depicted in bold red, while unplanted edges are shown in solid blue. The root vertex is denoted by $\root$. A planted child is labeled by appending $\tilde{1}$ or $\tilde{2}$ to its parent's label, while an unplanted child's label ends with an integer corresponding to its order of arrival in the Poisson process.}
    \label{fig:plantedH}
    \end{figure*}
    Next, we define the infinite tree $(\ell^S_\infty, T^S_\infty)$, illustrated in Figure~\ref{fig:plantedS}. The key insight is that each connected component of this infinite tree after removing the unplanted edges is a skeleton tree, as characterized in~\cite{Aldous1991Asymptotic, Grimmett1980Random}. Let $\root$ be the root vertex. The root has one planted child labeled $\tilde{0}$ and an additional $k \ge 0$ planted children labeled $\{\tilde{1}, \tilde{2}, \ldots, \tilde{k}\}$, where $k$ is drawn from a Poisson($1$) distribution. 
     In addition, the root is connected to an infinite sequence of unplanted vertices labeled $\{1, 2, 3, \ldots\}$. The weights of the planted edges (connecting $\root$ to its planted children) are i.i.d.\ samples from the distribution $P$, while the weights of the unplanted edges (connecting $\root$ to $\{1, 2, 3, \ldots\}$) are determined by the arrival times of a rate-$1$ Poisson process.
    
    As in the previous case, the descendant distribution for any non-root vertex depends on its type. For vertices whose labels end with a non-tilded integer or $\tilde{0}$, the descendants have the same distribution as those of the root. For all other vertices, the distribution is similar to that of the root but excludes a child labeled $\tilde{0}$. Figure~\ref{fig:plantedS} illustrates the structure of $(\ell^S_\infty, T^S_\infty)$ and the labeling rules: a planted child's label is formed by appending a tilded integer to its parent's label, whereas an unplanted child's label is determined by appending a positive integer corresponding to its order in the Poisson process.

    \begin{figure*}[h!]
    \centering
    \includegraphics[width=0.99\textwidth]{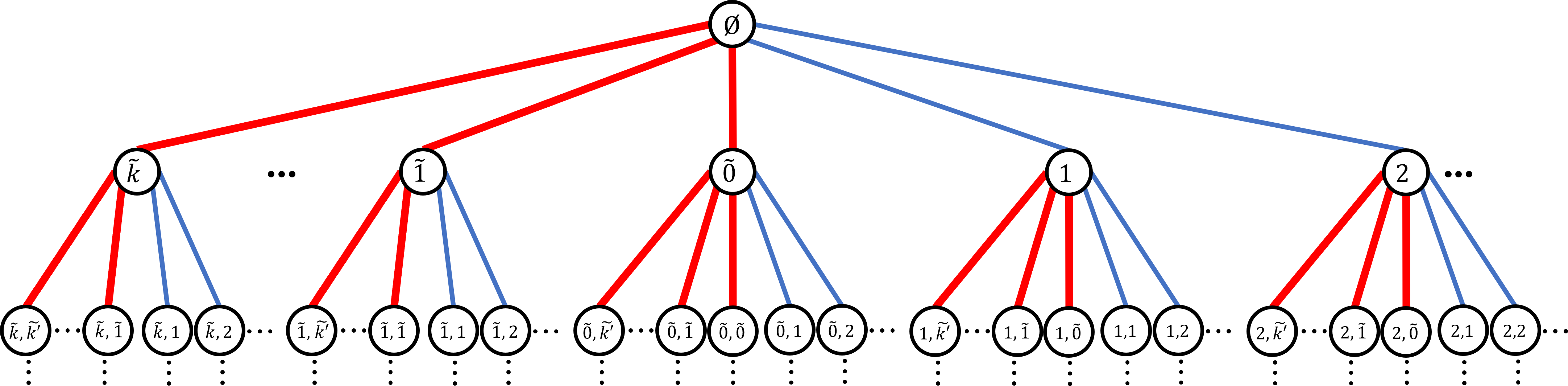}
    \caption{Structure of $(\ell^S_\infty, T^S_\infty)$. Planted edges are depicted in bold red, while unplanted edges are shown in solid blue. The root vertex is denoted by $\root$. A planted child is labeled with a tilded integer appended to its parent's label, while an unplanted child is labeled with its arrival order in the Poisson process. $k\geq 0$ and $k'\geq 0$ are independently drawn from a Poisson distribution with parameter 1, and abusing the notation, the same symbols are reused for the descendants of different nodes.
    }
    \label{fig:plantedS}
    \end{figure*}

    For a uniformly chosen vertex, short cycles (measured in total edge weight) vanish asymptotically, causing the graph to converge locally to a tree. This observation underpins the concept of local weak convergence. We conclude the first step of the proof by stating that the finite graph models converge in the local weak sense to their infinite-tree counterparts. A detailed proof can be found in Appendix~\ref{app:localweakconv}.
    
    \begin{lemma}[Local Weak Convergence of Finite Graphs] \label{thm:localweakconv}
        The sequence $\{(\ell^H_n, G^H_n)\}_n$ converges locally weakly to $(\ell^H_\infty, T^H_\infty, \root)$, and $\{(\ell^S_n, G^S_n)\}_n$ converges locally weakly to $(\ell^S_\infty, T^S_\infty, \root)$.
    \end{lemma}

    \subsubsection*{Step 2: Convergence of the Minimum Spanning Tree.}
    Given that $\{(\ell^\square_n, G^\square_n)\}_n$ converges locally weakly to $(\ell^\square_\infty, T^\square_\infty, \root)$, it is natural to expect that the corresponding spanning trees $M_n^\square$ also converge (in the local weak sense) to a subgraph $M_\infty^\square$ of $T^\square_\infty$. Because each $M_n^\square$ spans all vertices of $G^\square_n$, the subgraph $M_\infty^\square$ should similarly span every vertex of $T^\square_\infty$. Moreover, since $T^\square_\infty$ is a tree, any subgraph of it must be a forest. Hence, $M_\infty^\square$ is a forest spanning $T^\square_\infty$.

    To characterize $M_\infty^\square$, consider any edge $e_n = \{u_n, v_n\}\notin M_n^\square$. 
    If all edges in $G_n$ with weights at least $\ell_n^\square(e_n)$ are removed, there still exists a path connecting $u_n$ and $v_n$ in the truncated graph. 
    This observation is crucial for understanding $M_\infty^\square$: since $T_\infty^\square$ is acyclic, the path between $u_n$ and $v_n$ must appear as paths extending to infinity in the limiting object.

    More precisely, suppose the local structure around $e_n \notin M_n^\square$ converges to the local structure around an edge $e = \{u, v\}$ in $T_\infty^\square$ with $e \notin M_\infty^\square$. 
    When all edges whose weights are at least $\ell_\infty^\square(e)$ are removed, the resulting subtrees rooted at $u$ and $v$ must extend to infinite depth. Note that shifting from local neighborhoods of vertices to local neighborhoods of edges involves a change of measure, as detailed in~\cite{MoharramiMoorXu2021Planted}.

    This observation underlies the definition of the minimal spanning forest for infinite graphs, as presented in~\cite{Steele2002Minimal}. Specifically, the minimal spanning forest $M_\infty^\square$ of $T_\infty^\square$ is a spanning forest that excludes an edge $e = (u,v)$ if and only if both $\mathcal{C}(u; T_\infty^\square, e)$ and $\mathcal{C}(v; T_\infty^\square, e)$ are infinite. Here, $\mathcal{C}(u; T_\infty^\square, e)$ denotes the connected component containing $u$ in the subgraph of $T_\infty^\square$ obtained by removing all edges whose length is at least $\ell_\infty^\square(e)$.

    Observe that for any vertex in $T_\infty^\square$, the adjacent edge with the smallest weight belongs to $M_\infty^\square$, ensuring that $M_\infty^\square$ is indeed a spanning forest.  Moreover, it can be shown that none of the connected components of $M_\infty^\square$ are finite, and with probability one, $M_\infty^\square$ is unique.

    Given the above characterization of $M_\infty^\square$, we state our main result on the joint local weak convergence of $\{(\ell^\square_n, G^\square_n,M^\square_n)\}_n$ to  $(\ell^\square_\infty, T^\square_\infty,M^\square_\infty, \root)$. This follows as a special case of the results in~\cite{AldousSteel2001Objective, Steele2002Minimal}. Further details on properties of $M_\infty^\square$ and the joint local weak convergence are provided in Appendix~\ref{app:localweakconv_tree}.

    \begin{theorem}[Joint Local Weak Convergence of Finite Graphs and Their MSTs] \label{thm:mstconv}
    The sequence $\{(\ell^H_n, G^H_n,M_n^H)\}_n$ converges jointly locally weakly to $(\ell^H_\infty, T^H_\infty, M_\infty^H,\root)$, and $\{(\ell^S_n, G^S_n, M_n^S)\}_n$ converges jointly locally weakly to $(\ell^S_\infty, T^S_\infty, M_\infty^S, \root)$.
    \end{theorem}

    \subsubsection*{Step 3: Detection Probability}

    Given Theorem \ref{thm:mstconv}, we characterize the asymptotic overlap by calculating the probability that a planted edge is included in the minimum spanning forest:
    \begin{align*}
        \lim_{n \to \infty} \mathbb{E}[\text{overlap}(M^\square_n, M^{\square,*}_n)] 
        &= \lim_{n \to \infty} \mathbb{P}(\text{a planted edge $e$} \in M^\square_n) \\
        &= \mathbb{P}(\text{a planted edge $e$} \in M^\square_\infty),
    \end{align*}
    where the last equality follows directly from Theorem \ref{thm:mstconv}. By definition, for a planted edge $e = (u, v)$ to belong to $M^\square_\infty$, at least one of $\mathcal{C}(u; T_\infty^\square, e)$ or $\mathcal{C}(v; T_\infty^\square, e)$ must be finite. In other words, once we remove all edges whose weights greater than or equal to $\ell_\infty^\square(e)$, at least one of the two branching processes originating from $u$ or $v$ must eventually go extinct.    
    
    To compute the extinction probabilities for the two resulting branching processes, we analyze $T_\infty^\square$ from the perspectives of $u$ and $v$, treating it as a doubly rooted tree. Figure~\ref{fig:plantedH-doublyroot} illustrates $T_\infty^H$ from the perspective of a planted edge, emphasizing its symmetric structure. In contrast, Figure~\ref{fig:plantedS-doublyroot} depicts $T_\infty^S$ from the same perspective, showing a planted subgraph that is bounded on one side (the right side in the figure) with probability $1$ and unbounded on the other side (the left side in the figure). Notably, a vertex with label $\tilde{0}$ appears exclusively on the left side of~\ref{fig:plantedS-doublyroot}, forming an infinite path that begins at $(-\root, -\tilde{0}).$

    \begin{figure*}[h!]
        \centering
        \begin{subfigure}[b]{0.48\textwidth}
            \centering
            \includegraphics[width=\textwidth]{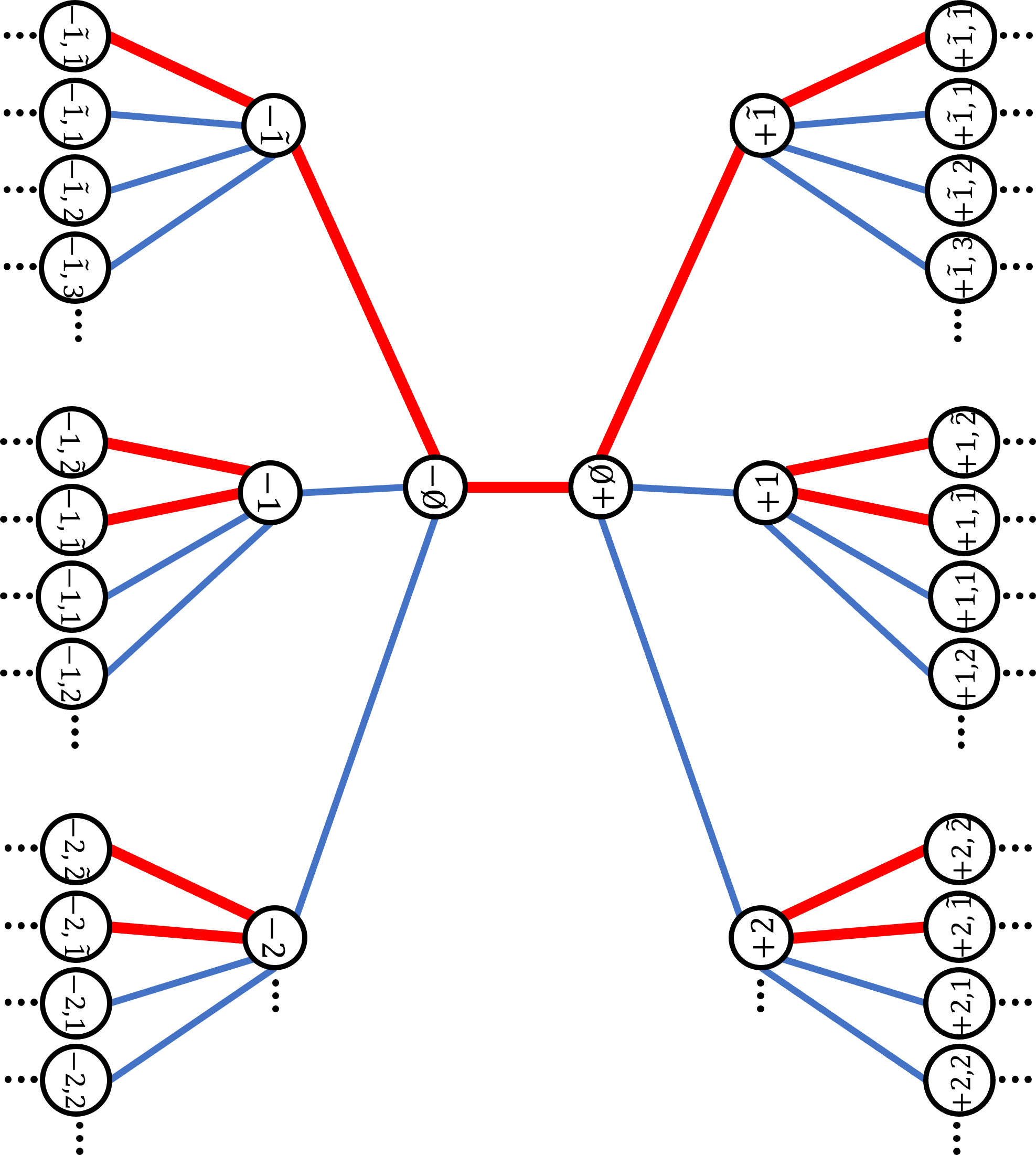}
            \caption{Doubly-rooted version of $T_\infty^H$}
            \label{fig:plantedH-doublyroot}
        \end{subfigure}
        \hfill
        \begin{subfigure}[b]{0.48\textwidth}
            \centering
            \includegraphics[width=\textwidth]{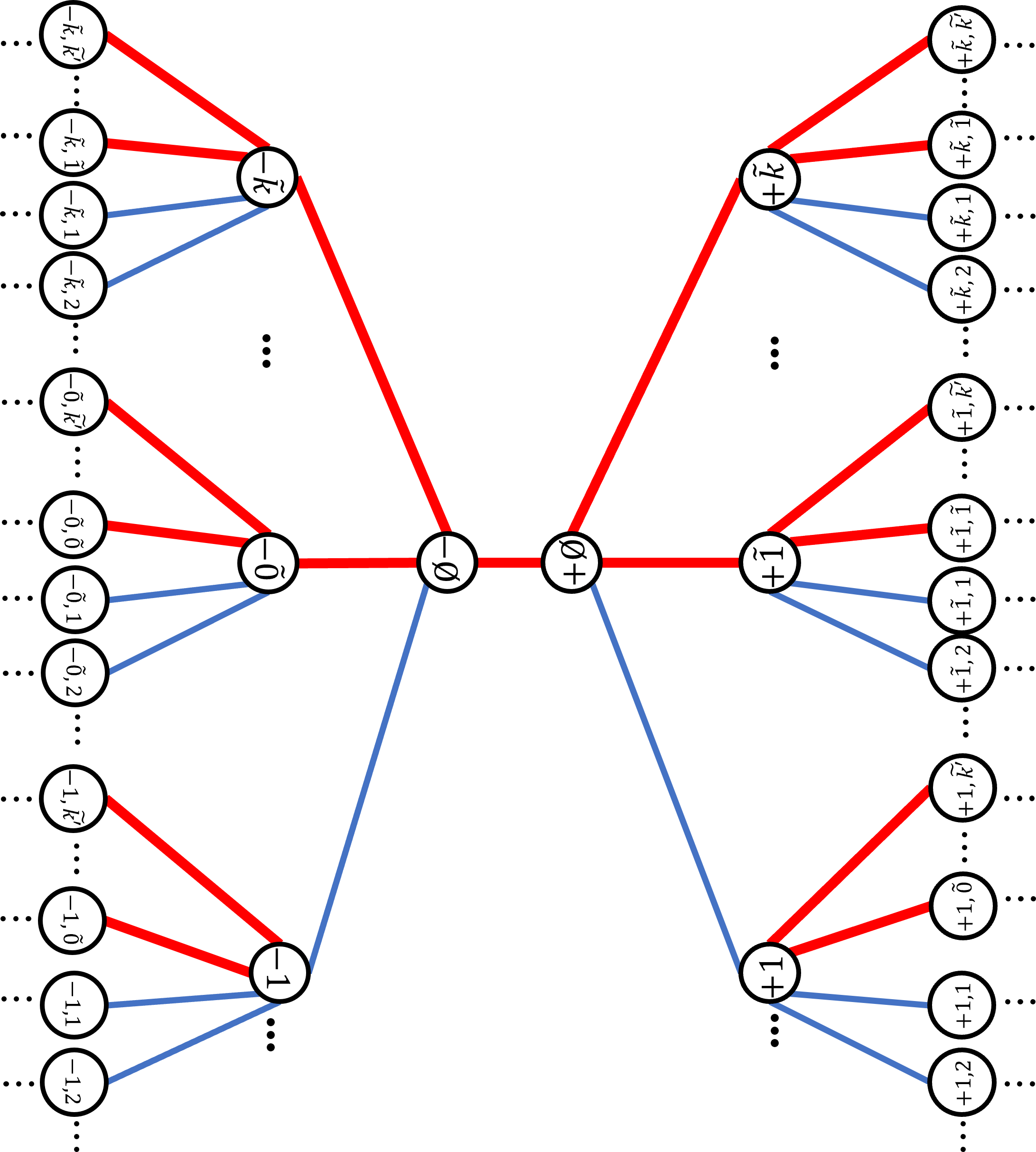}
            \caption{Doubly-rooted version of $T_\infty^S$}
            \label{fig:plantedS-doublyroot}
        \end{subfigure}
        \caption{The ``+'' and ``-'' signs are used to distinguish the labels on different sides.}
        \label{fig:planted-doublyrooted}
    \end{figure*}
    
    It is worth noting that the descendant distributions for the root vertices in both figures are essentially the same, except for the presence of a vertex labeled $\tilde{0}$ in Figure~\ref{fig:plantedS-doublyroot}. More specifically, in Figure~\ref{fig:plantedH-doublyroot}, both root vertices $-\root$ and $+\root$ have exactly one planted descendant. In Figure~\ref{fig:plantedS-doublyroot} however, the vertex $-\root$ has $\tilde{X}_1 + 1$ planted descendants, while $+\root$ has $\tilde{X}_2$ planted descendants, where $\tilde{X}_1$ and $\tilde{X}_2$ have sized-biased distributions. Notice that the size-biased distribution of a Poisson distribution remains a Poisson distribution. In both cases, each root vertex also has infinitely many unplanted descendants, whose edge weights are determined by the arrival times of a rate‐1 Poisson process.

    Suppose that $\ell_\infty(-\root, +\root) = s$, and let $p^\square_+(s)$ denote the extinction probability of the branching process on the ``+'' side, and $p^\square_-(s)$ denote the extinction probability of the branching process on the ``-'' side, after removing all edges with edge weights greater than or equal to $s$. Notice that $p^H_+(s) = p^H_-(s)$, while $p^S_+(s) \neq p^S_-(s)$. 
    
    Let us first look at the uniform Hamiltonian path model, illustrated in Figure~\ref{fig:plantedH-doublyroot}. Consider the branching process emanating from the ``-'' side after removing all edges with weights greater than or equal to $s$. For this process to become extinct, the subtree rooted at the planted edge $(-\root, -\tilde{1})$ and all the subtrees rooted at $-1, -2, \ldots, -r$ must be extinct, where $r$ has Poisson distribution with parameter $s$. Denote by $p(s)$ the extinction probability of the subtree rooted at $\root$, and by $q(s)$ that of the subtree rooted at $-1$. Using the moment-generating function of the Poisson distribution, we derive the following equations: \begin{align*}
    p(s) & = \mathbb{E}_{r \sim \text{Pois}(s)}\left[ 
    \big(q(s)\big)^r \right] \times (1-F(s)+F(s)p(s)) \\
    & =
    \exp(-s(1-q(s)))(1-F(s)+F(s)p(s)).
    \end{align*} 
    Similarly, considering the subtree rooted at $-1$, we obtain
    $
    q(s)=\exp(-s(1-q(s)))(1-F(s)+F(s)p(s))^2
    .
    $
    Thus, we have obtained the fixed-point equations~\eqref{eq:fixed_point_HP}.

    Next, let us look at the uniform spanning tree model, illustrated in Figure~\ref{fig:plantedS-doublyroot}. Consider the branching process on the ``-'' side after removing all edges with edge weights greater than or equal to $s$. For this process to become extinct, all the following subtrees must be extinct: $(1)$ subtrees rooted at $-1, -2, \ldots, -r$ where $r$ has Poisson distribution with parameter $s$, $(2)$ all the subtrees rooted at planted edges $(-\root, -\tilde{1}), (-\root, -\tilde{2}), \ldots, ( -\root, -\tilde{k})$  where $k$ has Poisson distribution with parameter $1$, and $(3)$ the subtree rooted at the planted edge $(-\root,-\tilde{0})$. Denote the extinction probability for the tree rooted at $-\root$ and $-\tilde{1}$   as $p_U(s)$ and $p_B(s)$, respectively. Then we arrive at 
    \begin{align*}
    p_U(s) & = 
    \mathbb{E}_{r \sim \text{Pois}(s)}\left[ 
    \big(p_U(s)\big)^r \right] 
    \times 
    \mathbb{E}_{k \sim \text{Pois}(1)}\left[ 
    \big(p_B(s)\big)^k \right]
    \times
    (1-F(s)+F(s)p_U(s))  \\
    & = \exp(-s(1-p_U(s)) \exp(- (1-p_B(s)) (1-F(s)+F(s)p_U(s)).
    \end{align*}
    Similarly, considering the branching process on the ``+'' side, we obtain 
    $
    p_B(s)=\exp(-s(1-p_U(s))\exp(-(1-p_B(s))).
    $
    Hence, we have obtained the fixed-point equations~\eqref{eq:fixed_point_US}.

    In either model, the extinction probability can be obtained by iteratively applying the corresponding fixed-point equations—\eqref{eq:fixed_point_US} for the uniform spanning tree model or \eqref{eq:fixed_point_HP} for the uniform Hamiltonian path model—starting from the all-zero function. After $t$ iterations, the resulting value represents the probability that the corresponding branching process dies out by generation~$t$, meaning that no vertices survive beyond $t$ edges from the root. Since the probability of dying out by generation~$t$ is non-decreasing in $t$, it follows that as $t \to \infty$, the iteration converges to the smallest fixed point of the equations
    \eqref{eq:fixed_point_US} or \eqref{eq:fixed_point_HP}.

    Given $\ell_\infty(-\root, +\root) = s$, the probability of the event $\{-\root, +\root\} \notin M^\square_\infty$ equals $(1 - p_{-}(s))(1 - p_{+}(s))$. Therefore, the probability that a planted edge $e$ belongs to $M^\square_\infty$ is given by:
    \begin{align*}
        \mathbb{P}(\text{a planted edge } e \in M^\square_\infty) = \int_0^\infty \big(1 - (1 - p_{-}(s))(1 - p_{+}(s))\big) \, dF(s).
    \end{align*}

To determine the asymptotic expected weight of $M_n^\square$, observe that 
    \begin{align*}
        \mathbb{E} \left[ \frac{1}{n-1} \sum_{e\in E_n^\square} \ell_n(e) \, \mathbf{1}_{\{e \in M_n^\square\}} \right] &= 
        \mathbb{E} \left[ \frac{1}{2(n-1)} \sum_{ v \in V_n^\square} \sum_{u  \in V_n^\square: u \neq v} \ell^\square_n(u,v) \, \mathbf{1}_{\{(u,v) \in M_n^\square \}} \right] \\
        &= \mathbb{E} \left[ \frac{1}{2} \sum_{v \in V_n^\square} \ell^\square_n(u_1,v) \, \mathbf{1}_{\{(u_1,v) \in M_n^\square \}} \right]
    \end{align*}
    where $u_1 $ is any fixed vertex of $V_n^\square$ that is not equal to $v$. By Theorem~\ref{thm:mstconv}, we have
    \begin{align*}
        \sum_{v \in V_n^\square} \ell^\square_n(u_1,v) \, \mathbf{1}_{\{(u_1,v) \in M_n^\square \}} 
        \xrightarrow{d} 
        \sum_{ v: (\root,v) \in E_\infty^\square }  
        \ell^\square_\infty(\root,v) \, \mathbf{1}_{\{(\root,v) \in M_\infty^\square \}}.
    \end{align*}
    Under the unplanted models, the left-hand side is shown to be uniformly integrable~\cite{AldousSteel2001Objective}. 
    Moreover, it is straightforward to verify that the portion of the summation involving edges in $M_n^\square \cap M_n^{\square,*}$ is also uniformly integrable.
    Hence, the left-hand side is also uniformly integrable under the planted model. It follows that 
    \begin{align*}
        \lim_{n\to\infty} \mathbb{E} \left[ \frac{1}{2} \sum_{v \in V_n^\square} \ell^\square_n(u_1,v) \, 
        \mathbf{1}_{\{(u_1,v) \in M_n^\square \}}\right] 
        =  \frac{1}{2} \mathbb{E} \left[  \sum_{v: (\root,v) \in E_\infty^\square} \ell^\square_\infty(\root,v) \, 
        \mathbf{1}_{\{(\root,v) \in M_\infty^\square \}} \right].
    \end{align*}
    Therefore, to derive the limits in Theorem~\ref{thm:asymcost}, we must compute the right-hand side of the above equation. Observe that this expression can be written as
    \begin{align*}
        \frac{1}{2}\mathbb{E} \left[  \sum_{v \in V_\infty^\square} \ell^\square_\infty(\root,v) \, 
        \mathbf{1}_{\{(\root,v) \in M_\infty^\square \cap M_\infty^{\square,*} \}} \right]
        + \frac{1}{2}\mathbb{E} \left[  \sum_{v \in V_\infty^\square} \ell^\square_\infty(\root,v) \, 
        \mathbf{1}_{\{(\root,v) \in M_\infty^\square \setminus M_\infty^{\square,*} \}} \right].
    \end{align*}
    Condition on a planted edge $e$ with $\root \in e$ and $\ell_\infty(e) = s$. In that case, the resulting graph has the same structure depicted in Figure~\ref{fig:planted-doublyrooted}, viewed from $e$ as a doubly rooted tree. As noted earlier, the probability of the event $\{-\root, +\root\} \notin M^\square_\infty$ equals $(1 - p_{-}(s))(1 - p_{+}(s))$. Hence, noting that the expected number planted children of $\root$ in both models is $2$, we have
    \begin{align*}
        \mathbb{E} \left[  \sum_{v \in V_\infty^\square} \ell^\square_\infty(\root,v) \, 
        \mathbf{1}_{\{(\root,v) \in M_\infty^\square \cap M_\infty^{\square,*} \}} \right] = 2\int_0^\infty s \big(1 - (1 - p_{-}(s))(1 - p_{+}(s))\big) \, dF(s).
    \end{align*}
   On the other hand, if we condition on an unplanted edge $e$ with $\root \in e$ and $\ell_\infty(e) = s$, then the resulting graph has a symmetric structure when viewed from $e$ as a doubly rooted graph: the subtree on each side is a copy of the original rooted tree, as depicted in Figures~\ref{fig:plantedH} and~\ref{fig:plantedS}; indeed, if we condition on an arrival at the point $s$ in a Poisson process, the set of remaining points still forms a Poisson process with the same rate. Finally, observe that the ``distribution" of the weight of such an edge $e$ is simply uniform on $\mathbb{R}_+$, i.e., it follows the Lebesgue measure on the positive reals. Following the same argument for the probability of the event $\{-\root, +\root\}\notin M_\infty^\square$, we have
   \begin{align*}
       &\mathbb{E} \left[  \sum_{v \in V_\infty^S} \ell^S_\infty(\root,v) \, 
        \mathbf{1}_{\{(\root,v) \in M_\infty^S \setminus M_\infty^{S,*} \}} \right] = \int_0^\infty s \big(1 - (1 - p_{U}(s))^2\big) \, ds,\\
        &\mathbb{E} \left[  \sum_{v \in V_\infty^H} \ell^H_\infty(\root,v) \, 
        \mathbf{1}_{\{(\root,v) \in M_\infty^H \setminus M_\infty^{H,*} \}} \right] = \int_0^\infty s \big(1 - (1 - q(s))^2\big) \, ds.
   \end{align*}
    
\section{Conclusions and Open Problems} 
In this paper, we provide an exact characterization of the asymptotic overlap between the minimum-weight spanning tree and the planted spanning tree. Furthermore, we derive the asymptotic value of the mean weight of MST, which generalizes the famous Frieze's $\zeta(3)$ result to the planted model. Finally, we design an efficient test based on the MST weight to distinguish the planted model from the unplanted model with vanishing testing error. Our analysis extends local weak convergence theory to describe the asymptotic local structure of the planted model. A promising direction for future work is to determine the information-theoretically optimal overlap. Another interesting avenue is to explore other spanning structures, such as spanning regular graphs as studied in~\cite{sicuro2021planted} and~\cite{gaudio2025all}. Finally, we note that  our results on the planted Hamiltonian path apply directly to the planted Hamiltonian cycle~\citep{bagaria2020hidden}, as both paths and cycles share the same weak limit.

\section*{Acknowledgment}
J.~Xu is supported in part by an NSF CAREER award CCF-2144593.

\bibliographystyle{plainnat}
\bibliography{bibliography}
\newpage
\appendix
\section{Data Points for the Plots in Figure~\ref{fig:planted-simuvsth}}\label{app:table}
\begin{table}[h!]
\centering
\begin{tabular}{r|cc|cc}
\toprule
\multicolumn{1}{c|}{$\mu$} & \multicolumn{2}{c|}{Uniform Spanning Tree} & \multicolumn{2}{c}{Uniform Hamiltonian Path} \\
\cmidrule(r){2-3} \cmidrule(l){4-5}
 & Overlap & Mean Weight & Overlap & Mean Weight \\
\midrule
  0.089667 & 0.913074 & 0.076304 & 0.906252 & 0.072267 \\
  0.311334 & 0.789809 & 0.213841 & 0.788253 & 0.200390 \\
  0.402602 & 0.752422	& 0.258459 & 0.753049 & 0.242444 \\
  4.418627 & 0.282288 & 0.852414 & 0.289205 & 0.840914 \\
  8.434651 & 0.178754 & 0.982680 & 0.182270 & 0.977254 \\
 12.852277 & 0.127802 & 1.045965 & 0.129785 & 1.043005 \\
 17.269904 & 0.099549 & 1.080776 & 0.100813 & 1.078922 \\
 21.285928 & 0.082916 & 1.101167 & 0.083818 & 1.099858 \\
 25.703554 & 0.070054 & 1.116881 & 0.070712 & 1.115934 \\
 30.121181 & 0.060652 & 1.128336 & 0.061152 & 1.127620 \\
 34.137205 & 0.054059 & 1.136353 & 0.054460 & 1.135781 \\
 38.554831 & 0.048287 & 1.143361 & 0.048610 & 1.142902 \\
 42.972458 & 0.043629 & 1.149008 & 0.043895 & 1.148632 \\
 46.988482 & 0.040112 & 1.153268 & 0.040338 & 1.152948 \\
 51.406108 & 0.036845 & 1.157220 & 0.037037 & 1.156950 \\
 55.823735 & 0.034071 & 1.160575 & 0.034235 & 1.160343 \\
 59.839759 & 0.031888 & 1.163212 & 0.032033 & 1.163009 \\
\bottomrule
\end{tabular}
\caption{Data points of uniform spanning tree vs uniform Hamiltonian path}\label{table:datapoints}
\end{table}

\section{Characterization of the Fixed Point Equations}\label{app:fixedpoint}
Notice that for both fixed-point equations \eqref{eq:fixed_point_HP} and \eqref{eq:fixed_point_US}, if either function equals $1$ at $s$, then the other function must also be $1$.
    \begin{enumerate}[label = (\arabic*)]
        \item \textbf{Uniform spanning tree}: For convenience, we restate the fixed-point equations~\eqref{eq:fixed_point_US} below:
        \begin{align*}
            p_B(s) &= \exp\big(-s(1 - p_U(s))\big) \exp\big(-(1 - p_B(s))\big),\\
            p_U(s) &= \exp\big(-s(1 - p_U(s))\big) \exp\big(-(1 - p_B(s))\big) \big(1 - F(s) + F(s)p_U(s)\big).
        \end{align*}
        Note that 
        $$
        \frac{ p_U(s)}{p_B(s)} = 1- F(s) +F(s) p_U(s), 
        $$
        and thus $p_U(s) = \frac{p_B(s) (1-F(s)) }{1-F(s) p_B(s) }$. Plugging it back to the expression of $p_B(s)$ yields 
        $$
        p_B(s) 
        = \exp\left(-s  \frac{1-p_B(s) }{1-F(s) p_B(s)} \right) \exp\left(-(1 - p_B(s))\right).
        $$
        For $s \ge 0$, define the function $\phi_s:[0,1]\to[0,1]$ as follows:
        $$
        \phi_s(x) = \exp\left(- (1-x) \left( \frac{s}{1-F(s) x} +1 \right)\right).
        $$
        Then $\phi_s(0)= \exp(-(1+s))>0$ and $\phi_s(1) = 1$. Moreover, 
        $$
        \phi_s'(x)=\phi_s(x) \left( 1+ \frac{s (1-F(s)) }{(1-F(s) x)^2} \right)  >0
        $$
        and 
        $$
        \phi_s''(x) = \phi_s'(x) \left( 1+ \frac{s (1-F(s)) }{(1-F(s) x)^2} \right) + 2 \phi_s(x)  \frac{s \left(1-F(s) \right) F(s) }{(1-F(s) x)^3} >0.
        $$
        Therefore, $\phi_s(x)$ is strictly increasing and convex in $[0,1]$. In addition, $\phi_s'(0)=\exp(-1-s)(1+s(1-F(s))) <1$, and $\phi_s'(1)=1+ \frac{s}{1-F(s)} \ge 1$, where the inequality is strict when $s>0.$ It follows that $\phi_s(x)$ has a unique fixed point $x^*(s)$  in $(0,1)$ for $s>0$ and $\phi_0(x)$
        a unique fixed point $x^*(0)=1$. Therefore, $p_B(s)=x^*(s)$.
        \item \textbf{Uniform Hamiltonian path}: For convenience, we restate the fixed-point equations~\eqref{eq:fixed_point_HP} below:
        \begin{align*}
            p(s) &= \exp\big(-s(1 - q(s))\big) \big(1 - F(s) + F(s)p(s)\big),\\
            q(s) &= \exp\big(-s(1 - q(s))\big) \big(1 - F(s) + F(s)p(s)\big)^2. 
        \end{align*}
        Note that 
        $$
        q(s) = p(s) \big(1 - F(s) + F(s)p(s)\big).
        $$
        Plugging the above back to the expression of $p(s)$ yields 
        $$
        p(s) = \exp\left(-s \left( 1 -p(s) \left(1 - F(s) + F(s)p(s)\right)  \right) \right) \left(1 - F(s) + F(s)p(s)\right).
        $$
        For $s \ge 0$, define the function $\phi_s:[0,1]\to[0,1]$ as below:
        $$
        \phi_s(x) = \exp\left(-s \left( 1 -x \left(1 - F(s) + F(s) x\right)  \right) \right) \left(1 - F(s) + F(s)x \right).
        $$     
        It follows that $\phi_s(0)=e^{-s}(1-F(s)) \ge  0$ with the inequality to be strict for $F(s)<1$ and $\phi_s(1)=1$. Moreover,
        $$
        \phi_s'(x) = s(1+F(s)) \phi_s(x) + 
        \exp\left(-s \left( 1 -x \left(1 - F(s) + F(s) x\right)  \right) \right) F(s) \ge 0
        $$
        and 
        $$
        \phi_s''(x) = s(1+F(s)) \phi_s'(x)
        + \exp\left(-s \left( 1 -x \left(1 - F(s) + F(s) x\right)  \right) \right) F(s) s (1+F(s)) \ge 0 
        $$
        The above inequalities are strict for $s>0.$
        Thus, for $s>0$, $\phi_s(x)$ is strictly increasing and convex in $[0,1]$; for $s=0$, $\phi_0(x)=1-F(0)+F(0)x$. 
        Note that $\phi'_s(1)=s(1+F(s))+F(s).$
        Therefore, $\phi_s(x)$ has a unique fixed point $x^*(s)$ in $(0,1)$ if and only if 
        if $s>\frac{1-F(s)}{1+F(s)}$.     
        Hence, $p(s)=x^*(s)$ if $s>\frac{1-F(s)}{1+F(s)}$
        and $p(s)=1$ otherwise.
    \end{enumerate}

    \section{Proof of Theorem~\ref{thm:hypothesis_testing}}\label{sec:proof_hypothesis}
We prove that $w(M_n)$ concentrates on its mean via Talagrand's inequality. Observe that
$$
w(M_n) = f(W) \triangleq \frac{1}{n-1}   \cdot \min_{M \in \mathcal{M}} \langle W, M \rangle,
$$
where $W\in \mathbb{R}^{\binom{n}{2}}$ is the edge weight vector, $M \in \{0,1\}^{\binom{n}{2}}$ denotes the indicator vector of a spanning tree, and $\mathcal{M}$ denotes the set of all possible indicator vectors of spanning trees.

By definition, $f(W)$ is convex and Lipschitz with Lipschitz constant upper bounded by $1/\sqrt{n-1}$. Moreover, $W$ has independent entries. However, $W_e$ may not be bounded. Thus, to apply Talagrand's inequality, we truncate $W.$ Specifically, given a threshold $\tau_n=16\log n$, 
let $W'_e= \min\{W_e, \tau_n\}$ for every edge $e$. Moreover, define a truncated graph $G_n'$ that consists of edges $e$ for which $W_e \le \tau_n$. By definition, $f(W) \ge f(W')$
and the equality holds when
$G_n'$ is connected, 

Recall that we have assumed the probability density of $Q_n$  at the origin satisfies $\lim_{n\to\infty} nQ'_n(0)=1.$ More specifically, suppose that $Q_n$ has density $\frac{1}{n}\rho(x/n)$, where $\rho$ is a fixed probability density function on $\mathbb{R}$
and $\rho$ is continuous at $0$ with $\rho(0)=1.$ Then, 
% Hence, for any small $\epsilon_n>0$, there exists $\tau_n\equiv \tau_n(\epsilon)>0$
% such that 
$$
Q_n\left(W_e \le \tau_n \right) =\int_0^{\tau_n}
\frac{1}{n}\rho(x/n) dx 
\ge (1-o(1)) \frac{\tau_n}{n} \rho(0),
$$
where the last inequality holds for sufficiently large $n$ by the continuity of $\rho(x)$ at $0$ and $\tau_n=o(n)$. Furthermore, since $\tau_n \to \infty,$
$P(W_e \le \tau_n)=1-o(1)$. It is well known that
 Erd\H{o}s--R\'enyi random graph $\mathcal{G}(n,4\log n/n)$ 
is connected with probability at least $1-O(n^{-3})$. Therefore, by coupling, $G_n'$ is also connected with probability at least $1-O(n^{-3})$. It follows that 
\begin{align}
\prob\{f(W)=f(W')\} \ge 1- O(n^{-3}). \label{eq:f_equal}
\end{align}
Applying Talagrand’s concentration inequality for Lipschitz convex functions (see, e.g.~\cite[Theorem 2.1.13]{tao2012topics}) yields that 
\begin{align}
\prob\left\{ \left|f(W') - \expect{f(W')} \right| \ge \frac{t\tau_n}{\sqrt{n-1}} \right\}
\le C \exp(-c t^2) \label{eq:talagrand}
\end{align}
for some absolute constants $C,c>0$. 
% Moreover,
% $$
% \expect{f(W')}
% \ge \expect{f(W') \indc{f(W')=f(W)}}
% =\expect{f(W)}
% $$
Moreover, by the monotone convergence theorem,
$$
\lim_{n\to \infty} \expect_{\mathcal{H}_0}{f(W')}=\lim_{n\to \infty } \expect_{\mathcal{H}_0}{f(W)} = \zeta(3).
$$
Therefore, for any given $\epsilon>0$, $\expect_{\mathcal{H}_0}{f(W')} \ge \zeta(3)-\epsilon/2$ for all sufficiently large $n.$ It follows  that 
\begin{align*}
\prob_{\mathcal{H}_0} 
\left\{ f(W) \le \zeta(3) -\epsilon \right\}
& \le \prob_{\mathcal{H}_0} 
\left\{ f(W') \le \zeta(3) -\epsilon \right\} + \prob_{\mathcal{H}_0}\left\{ f(W) \neq f(W')\right\} \\
& \le  \prob_{\mathcal{H}_0} 
\left\{ f(W') \le \expect_{\mathcal{H}_0}{f(W')} -\epsilon/2 \right\} + O(n^{-3})\\
& \le \exp(-\Omega(\epsilon^2 n/\log^2(n))) +O(n^{-3}),
\end{align*}
where the last two inequalities follow from~\eqref{eq:f_equal} and~\eqref{eq:talagrand}, respectively.
Similarly, by the assumption
$\lim_{n\to\infty} \expect_{\mathcal{H}_1}{f(W)}\le \zeta(3)-2\epsilon$,
we have $\expect_{\mathcal{H}_1}{f(W)}\le \zeta(3)-3\epsilon/2$
for all sufficiently large $n$ and therefore,  
\begin{align*}
\prob_{\mathcal{H}_1} 
\left\{ f(W) \ge \zeta(3) -\epsilon \right\}
& \le \prob_{\mathcal{H}_1} 
\left\{ f(W') \ge \zeta(3) -\epsilon \right\} + \prob_{\mathcal{H}_1}\left\{ f(W) \neq f(W')\right\} \\
& \le  \prob_{\mathcal{H}_1} 
\left\{ f(W') \ge \expect_{\mathcal{H}_1}{f(W')} +\zeta(3)- \expect_{\mathcal{H}_1}{f(W)} -\epsilon \right\} + O(n^{-3})\\
& \le \exp(-\Omega(\epsilon^2 n/\log^2(n))) + O(n^{-3}),
\end{align*}
where we used the fact that $\expect_{\mathcal{H}_1}{f(W')}\le \expect_{\mathcal{H}_1}{f(W)}$. 
Combining the last two displayed equations yields that
$$
\prob_{\mathcal{H}_0} 
\left\{ f(W) \le \zeta(3) -\epsilon \right\}+ \prob_{\mathcal{H}_1} 
\left\{ f(W) \ge \zeta(3) -\epsilon \right\} \le O(n^{-3}),
$$
concluding the proof.

\section{Local Weak Convergence Framework}\label{app:localweakframework}
In this section, we formally define the framework of local weak convergence. We outline only the key definitions necessary for our derivations; for a more comprehensive and technical exposition, we refer the reader to~\cite{MoharramiMoorXu2021Planted,Van2024Random}.

A \emph{weighted graph} $(\ell, G)$ consists of a graph $G = (V, E)$ with vertex set $V$ and edge set $E$, together with a weight function $\ell: E \to \mathbb{R}_{\geq 0}$ that assigns nonnegative weights to each edge. The weight function $\ell$ induces a metric $d_{\ell}$ on $V$, where the distance between two vertices is the infimum of the sums of edge weights over all paths connecting them.

A \emph{rooted weighted graph} $(\ell, G,\root)$, is a weighted graph $(\ell, G)$ together with a specific vertex $\root \in V$ designated as the root. All rooted weighted graph considered here are \emph{locally finite}, implying that for any $\rho > 0$, the number of vertices in the $\rho$-neighborhood of the root (i.e., the set of vertices at distance at most $\rho$ from $\root$) is finite. 

To compare rooted weighted graphs, we consider their \emph{isomorphism classes}, which identify graphs that are topologically equivalent and preserve edge weights. The set of all such isomorphism classes of rooted weighted graphs is denoted by $\mathcal{G}_\circ$. This space is equipped with a metric $d_\circ$ that captures the idea of \emph{closeness} between its members. Concretely, for two rooted weighted graphs $(\ell, G)$ and $(\ell', G')$, we say their distance is at most $1/(R+1)$ if there exists an isomorphism between the $R$-neighborhoods of their respective root vertices that preserves edges and ensures corresponding edge weights differ by at most $1/R$. 

The space $\mathcal{G}_\circ$ together with metric $d_\circ$ is a complete, and separable metric space, allowing us to employ the standard theory of weak convergence for probability measures on $\mathcal{G}_\circ$. More precisely, let $\mathcal{P}(\mathcal{G}_\circ)$ denote the set of probability measures on $\mathcal{G}_\circ$, endowed with the topology of weak convergence. We say that a sequence of probability measures $\{\eta_n\}_{n=1}^{\infty}$ in $\mathcal{P}(\mathcal{G}_\circ)$ \emph{converges weakly} to $\eta_\infty \in \mathcal{P}(\mathcal{G}_\circ)$, denoted by $\eta_n \xrightarrow{w} \eta_\infty$, if for every continuous and bounded function $f \colon \mathcal{G}_\circ \to \mathbb{R}$,
\begin{align*}
\int_{\mathcal{G}_\circ} f \, d\eta_n \longrightarrow \int_{\mathcal{G}_\circ} f \, d\eta_\infty     
\end{align*}

This notion of convergence is called \emph{local weak convergence} because the metric $d_\circ$ is defined locally in terms of the neighborhood of the root vertex, thereby capturing the asymptotic local structure of random weighted rooted graphs as viewed from the root. In the case of finite random weighted graphs $(\ell_n, G_n)$ without a designated root, we select a root vertex uniformly at random, and the local weak convergence is then defined with respect to the induced measure.

\begin{definition}[Local Weak Convergence]
    Consider a sequence of random finite weighted graphs $\{(\ell_n, G_n)\}_n$, where $G_n$ has $n$ vertices. Let $\root_n$ denote a randomly selected vertex from $G_n$ as the root vertex. Let $U_n \in \mathcal{P}(\mathcal{G}_\circ)$ denote the probability measure associated with $(\ell_n, G_n,\root_n)$. We say $(\ell_n, G_n)$ converges locally weakly to a random infinite rooted tree $(\ell_\infty, T_\infty, \root)\sim \eta$ if $U_n\xrightarrow{w}\eta$.
\end{definition}

\section{Local Weak Convergence of Finite Graph Models}\label{app:localweakconv}
To establish local weak convergence, we define an \emph{exploration process} that reveals the neighborhood of a fixed vertex $v$. Specifically, starting from $v$, the process uncovers all unseen planted neighbors, along with the $m$ closet new unplanted neighbors at each step. This procedure continues up to \emph{generation} $m$, where the generation of a node is defined by the number of edges between that node and the starting vertex in the explored subgraph. For the random finite weighted graph $(\ell^\square_n, G^\square_n)$, the exploration begins from a uniformly chosen vertex, while for the random infinite rooted tree $(\ell^\square_\infty, T^\square_\infty, \root)$, it starts from the designated root $\root$.

A key step is to show that, for any fixed $m$, with probability tending to $1$, there exists an \emph{edge-preserving isomorphism} between the explored neighborhoods in the finite and infinite models, and that their corresponding edge weights converge in distribution. Local weak convergence then follows by noting that for any fixed radius $\rho>0$, one can choose $m$ sufficiently large so that the $\rho$-neighborhood of the randomly selected vertex is contained within the explored subgraph. We now describe the exploration process in detail.

Consider a vertex $v$, which may be a randomly selected vertex in the finite graph or the designated root in the infinite tree. We begin with $v$ at \emph{generation} 0. At each subsequent generation $t = 1,2,\dots$:
\begin{enumerate}
    \item \textbf{Planted neighbors:} For every vertex $u$ discovered at generation $t-1$, reveal all of its planted neighbors that have not yet been seen.
    \item \textbf{Unplanted neighbors:} For each such vertex $u$, reveal the $m$ closest unplanted neighbors that have not yet been seen.
\end{enumerate}
This process continues until we reach generation $m$. At each stage, only newly encountered vertices are revealed, ensuring that no vertex is revisited in subsequent generations.

To establish an edge-preserving isomorphism between the explored neighborhoods, one must show that short cycles in terms of total edge weights vanish from the neighborhood of a randomly selected vertex in $(\ell^\square_n, G^\square_n)$ with high probability as $n \to \infty$, and that the descendant distribution of the finite graph model converges to that of the infinite rooted tree. Once such an isomorphism is obtained, the remainder of the argument becomes straightforward: the planted edges follow the same distribution~$P$ in both the finite and infinite models, and the distribution of the $m$ nearest unplanted neighbors converges to that of a rate-$1$ Poisson process. The disappearance of short cycles and the convergence of edge weight distributions were rigorously proved in~\cite{MoharramiMoorXu2021Planted} for the planted matching problem, and the same reasoning applies here with minimal modification.

Hence, our goal is to show that for any fixed $m>0$, as $n \to \infty$, the exploration process on the finite graph $(\ell^\square_n, G^\square_n)$ and on the infinite rooted tree $(\ell^\square_\infty, T^\square_\infty, \root)$ yields the same local structure. This is simpler to establish for the Hamiltonian model than for the uniform spanning tree model. Finally, the local weak convergence of the finite graph models follows directly from Lemma~\ref{lem:lwc_Hamiltonian} and Lemma~\ref{lem:lwc_Uniform}. 

\begin{lemma}\label{lem:lwc_Hamiltonian}
    With probability tending to $1$ as $n \to \infty$, there exists an edge-preserving isomorphism between the explored neighborhoods of $(\ell^H_n, G^H_n)$ and the infinite rooted tree $(\ell^H_\infty, T^H_\infty, \root)$.
\end{lemma}

\begin{proof}
    Let $f(m,n)$ be the number of vertices revealed by the exploration process on $(\ell^H_n, G^H_n)$. Observe that, with high probability, every vertex encountered in this process has exactly two planted neighbors. The reasons are twofold:
    \begin{enumerate}
        \item Except for two vertices in $G^H_n$, every vertex has exactly two planted neighbors. The probability that either of those two vertices appears in the exploration vanishes as $n \to \infty$.  
        \item With high probability, there are no short cycles, ensuring that each time we reveal planted neighbors, we add two new vertices.
    \end{enumerate}
    Consequently, as $n$ grows, $f(m,n)$ depends only on $m$, and the explored structure coincides with that of $(\ell^\square_\infty, T^\square_\infty, \root)$ with high probability. This implies the existence of an edge-preserving isomorphism, completing the proof.
\end{proof}

\begin{lemma}\label{lem:lwc_Uniform}
    With probability tending to $1$ as $n \to \infty$, there exists an edge-preserving isomorphism between the explored neighborhoods of $(\ell^S_n, G^S_n)$ and the infinite rooted tree $(\ell^S_\infty, T^S_\infty, \root)$.
\end{lemma}
\begin{proof}
\cite{Grimmett1980Random} employed a combinatorial approach to show that the uniform spanning tree converges locally to the skeleton tree. In particular, they demonstrated that, asymptotically, one can construct an edge-preserving isomorphism between the neighborhood of a randomly chosen vertex in the uniform spanning tree and the skeleton tree. We generalize their argument to prove the existence of such an edge-preserving isomorphism for our problem setting. The following lemmas provide essential asymptotic identities.
    \begin{lemma}\label{lem:noneighbor}
        Let $K_n = (V_n, E_n)$ denote a complete graph with $n$ vertices. Fix a set of vertices $\mathcal{V}_k = \{v^{(1)}, v^{(2)},\ldots,v^{(k)}\} \subset V_n$, and let $\mathcal{S}_n(k)$ denote the number of spanning trees that contain at least one edge between the vertices in $V_n$. Then, 
        \begin{align*}
            \lim_{n \to \infty} \frac{S_n(k)}{n^{n-2}} = 0
        \end{align*}
    \end{lemma}
    \begin{proof}[Proof of Lemma \ref{lem:noneighbor}]
        For any fixed vertices $ u, v \in V_n^S $, the number of spanning trees that include the edge $ e = \{u, v\} $ is $ 2n^{n-3} $. Hence,
        \begin{align*}
            S_n(k) \leq \binom{k}{2} \times 2n^{n-3}
        \end{align*}
    \end{proof}
    
    \begin{lemma}\label{lem:identity}
        Let $K_n = (V_n, E_n)$ denote a complete graph with $n$ vertices. Fix a set of vertices $\mathcal{V}_k = \{v^{(1)}, v^{(2)},\ldots,v^{(k)}\} \subset V_n$, and let $\mathcal{D}_n(d_1, d_2, \dots, d_k)$ denote the number of spanning trees in $K_n$ where the degree of vertex $v_i$, for $i \in \{1, 2, \dots, k\}$, is $d_i$, and none of the vertices in $\mathcal{V}_k$ are directly connected to each other. Then,
        \begin{align*}
            \lim_{n \to \infty} \frac{\mathcal{D}_n(d_1, d_2, \dots, d_k)}{n^{n-2}}
            = \frac{\mathrm{e}^{-k}}{(d_1 - 1)! (d_2 - 1)! \ldots (d_k - 1)!}.
        \end{align*}
    \end{lemma}
    
    \begin{proof}[Proof of Lemma \ref{lem:identity}]
        We use the Pr\"ufer code representation of spanning trees, where the number of times a vertex $v$ appears in the code is $\mathrm{degree}(v) - 1$. As a result, the joint degree distribution of $(v^{(1)}, v^{(2)},\ldots,v^{(k)})$ follows a multinomial distribution with parameters $n - 2$ and probabilities $p_1 = p_2 = \dots = p_k = \tfrac{1}{n}$ and $p_{k+1} = \tfrac{n - k}{n}$. More specifically,
        \begin{align*}
            &\prob\big(\mathrm{degree}(v^{(1)}) = d_1,\,
                         \mathrm{degree}(v^{(2)}) = d_2,\,
                         \dots,\,
                         \mathrm{degree}(v^{(k)}) = d_k \big) \\
            &\qquad = \binom{n - 2}{\,d_1 - 1,\,d_2 - 1,\,\dots,\,d_k - 1\,}
                      \biggl(\frac{1}{n}\biggr)^{\sum_{i=1}^k (d_i - 1)}
                      \biggl(\frac{n - k}{n}\biggr)^{\,n - 2 - \sum_{i=1}^k (d_i - 1)} \\
            &\qquad \xrightarrow[n\to\infty]{}
                    \frac{\mathrm{e}^{-k}}{(d_1 - 1)! (d_2 - 1)! \ldots (d_k - 1)!}.
        \end{align*}
        Finally, as $n \to \infty$, the relative number of spanning trees with the given degree sequence in which any pair of vertices among $v_1, \dots, v_k$ is directly connected converges to zero. Thus, the stated limit holds, completing the proof.
    \end{proof}
   
    Let $\mathcal{H}_t = (\mathcal{V}_t,\mathcal{E}_t)$ denote the subgraph explored up to generation $t$ in the exploration process for $(\ell^S_n, G^S_n)$.    
    Suppose $\mathcal{H}_t$ contains a total of $N$ vertices and $k$ planted components (connected components of $\mathcal{H}_t$ after removing the unplanted edges).
    Numbering the planted components arbitrarily, let $r_i$ denote the number of vertices at generation $t$  of planted component $i$. The vertices at generation $t$ of planted component $i$ are denoted by $v^{(i)}_j$ for $j \in \{1, 2, \dots, r_i\}$. Figure~\ref{fig:example} provides an example illustrating these associated numbers.
    
    \begin{figure*}[h!]     
    \centering
    \includegraphics[width=0.8\textwidth]{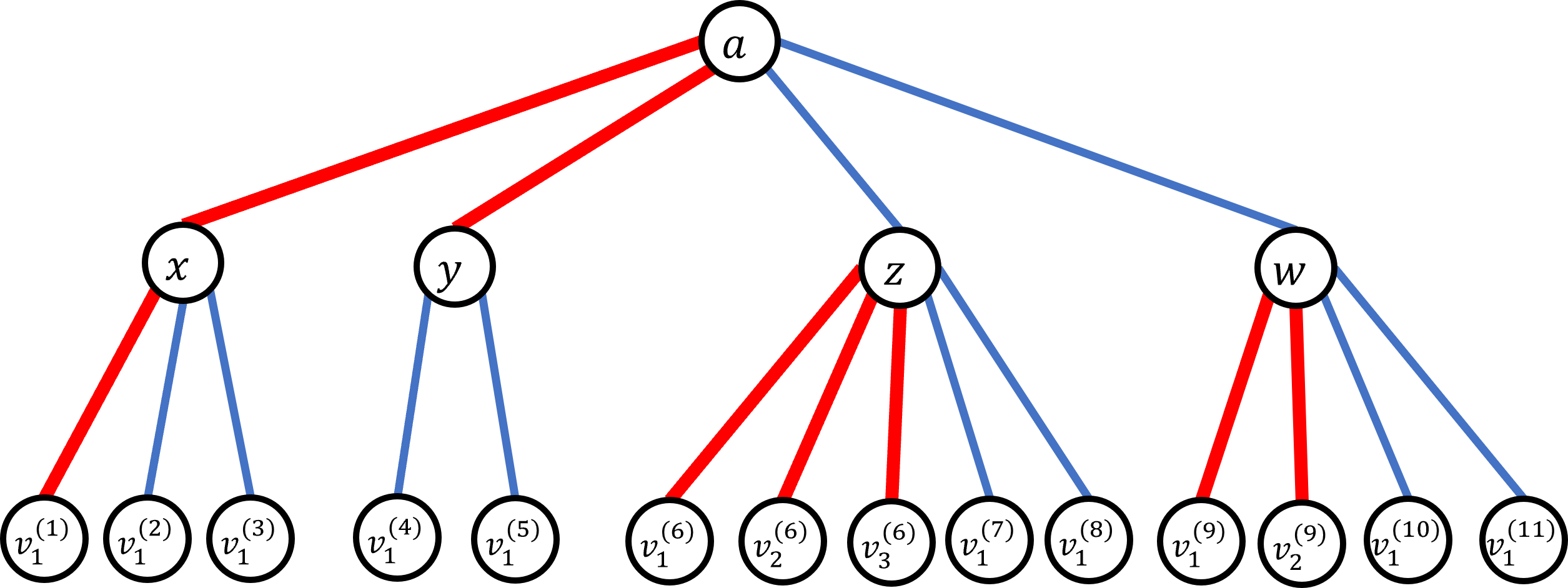}
    \caption{An example of a subgraph $\mathcal{H}_2$ with $m = 2$ is shown. The total number of vertices is $N = 20$, and the number of planted components is $k = 11$. Numbering these components from left to right, we have $r_1= r_2 = r_3 = r_4 = r_5 = 1$, $r_6 = 3$, $r_7 = r_8 = 1$, $r_9 = 2$, and $r_{10} = r_{11} = 1$.
    }
    \label{fig:example}
    \end{figure*}

    Our goal is to characterize the descendant distribution of vertices at generation $t$. For a vertex $v^{(i)}_j$, let $\mathrm{des}\big(v^{(i)}_j\big)$ denote the number of its descendants that appear at generation $t+1$ during the exploration process. Given a collection of nonnegative integers $d^{(i)}_j$ for $i=\{1,\dots,k\}$ and $j=\{1,\dots,r_i\}$, we are interested in 
    \begin{align*}
        \lim_{n \to \infty} \prob \left(\mathrm{des}\big(v^{(i)}_j\big) = d^{(i)}_j \,\text{ for all } i \in \{1,2,\dots,k\}        \text{ and } j \in \{1,2,\dots,r_i\} \mid \mathcal{H}_t\right).
    \end{align*}

    To compute the above conditional probability we use a counting argument. Note that this probability equals
    \begin{align*}
        \frac{\text{(number of spanning trees containing } \mathcal{H}_t \text{ with the given descendant distribution)}}{\text{(number of spanning trees containing } \mathcal{H}_t\text{)}}.
    \end{align*}

    To calculate the numerator, it is equivalent to replacing each planted component $i$ by a single vertex $v^{(i)}$, constructing a new complete graph on these vertices $\big\{v^{(1)}, v^{(2)}, \dots, v^{(k)}\big\}$ plus all remaining vertices in $V_n^S \setminus \mathcal{V}_t$, computing the quantity $\mathcal{D}_{n-N+k-2}(d_1, d_2, \dots, d_k)$ as defined in Lemma \ref{lem:identity}, where $d_i = \sum_{j=1}^{r_i} d^{(i)}_j,$ and then distributing the $d_i$ descendants of $v^{(i)}$ among the vertices $v^{(i)}_1,\dots,v^{(i)}_{r_i}$ according to their prescribed degrees. Consequently,
    \begin{align*}
        \textbf{Numerator} &= \mathcal{D}_{n-N+k}(d_1, d_2, \dots, d_k) \times 
        \prod_{i=1}^k \binom{d_i}{\,d^{(i)}_1,\,d^{(i)}_2,\dots,d^{(i)}_{r_i}} \\
        &\sim e^{-k} (n-N+k)^{n-N+k-2} \prod_{i=1}^k \left( d_i \prod_{j=1}^{r_i} \frac{1}{d^{(i)}_j!} \right)
    \end{align*}
    
    To compute the denominator, we use the Pr\"ufer code representation. Consider a Pr\"ufer code of length $n - N + k - 2$ constructed from a set of labels containing $n - N + \sum_{i=1}^k r_i$ distinct elements. These labels include all vertices in $V_n^S \setminus \mathcal{V}_t$ as well as the vertices $\{\,v_j^{(i)} : i \in \{1,2,\dots,k\},\; j \in \{1,2,\dots,r_i\}\}$. Each such labeling corresponds to a tree in which the labels in $\{v_1^{(i)}, v_2^{(i)}, \dots, v_{r_i}^{(i)}\}$ are replaced by a single delegate vertex $v^{(i)}$ for each $i \in \{1,2,\dots,k\}$.

    Notice that for each $i$, the vertices in $\{v_1^{(i)}, v_2^{(i)}, \dots, v_{r_i}^{(i)}\}$ appear collectively $d_i - 1$ times in the Pr\"ufer code, where $d_i = \sum_{j=1}^{r_i} d^{(i)}_j$ is the degree of the delegate vertex $v^{(i)}$. 
    Thus, the Pr\"ufer code indicates which neighbors of the delegate vertex $v^{(i)}$ are to be connected to the individual members of $\{v_1^{(i)}, v_2^{(i)}, \dots, v_{r_i}^{(i)}\}$, with the exception of one neighbor that does not appear directly in the code. (See Figure~\ref{fig:example2} for a detailed illustration.) There are $r_i$ possible choices for selecting this additional neighbor from the set $\{v_1^{(i)}, v_2^{(i)}, \dots, v_{r_i}^{(i)}\}$.
    
    Finally, as mentioned earlier, the relative frequency with which any two delegate vertices are directly connected is negligible. Consequently,
    \begin{align*}
        \textbf{Denominator} \sim (n - N + \sum_{i=1}^k r_i)^{\,n - N + k - 2} \times \prod_{i=1}^k r_i.
    \end{align*}

    \begin{figure*}[h!]     
    \centering
        \begin{subfigure}[b]{0.25\textwidth}
            \centering
            \includegraphics[width=\textwidth]{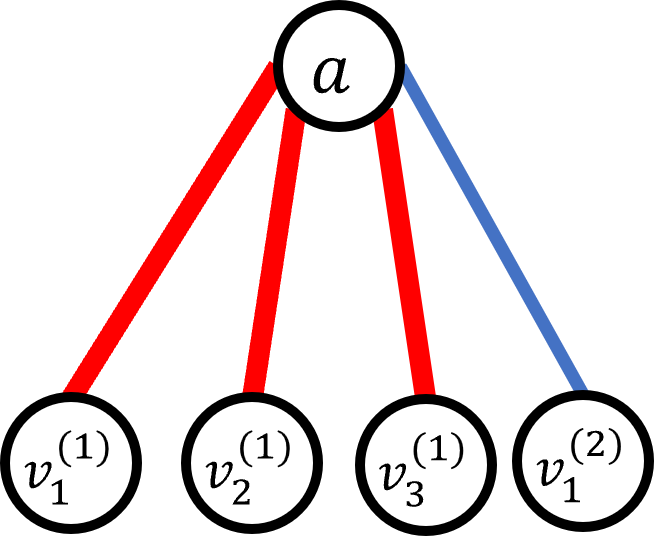}
            \caption{Subgraph $\mathcal{H}_1$}
        \end{subfigure}
        \hfill
        \begin{subfigure}[b]{0.5\textwidth}
            \centering
            \includegraphics[width=\textwidth]{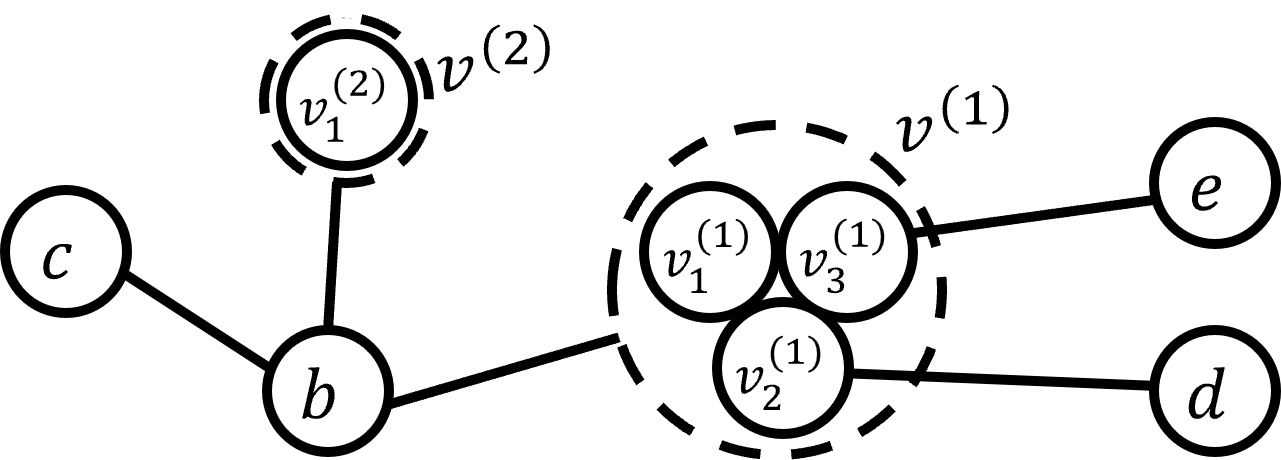}
            \caption{Tree given by the Pr\"ufer code $(b,v_2^{(1)},v_3^{(1)},b)$}
        \end{subfigure}
    \caption{An example of a tree resulting from a Pr\"ufer code of length $n - N + k - 2$ with $n - N + \sum_{i=1}^k r_i$ distinct labels is illustrated here. In this example, we have $n=9$, $N=5$, $t=1$, $k=2$, $r_1=3$, $r_2=1$, and $V_n^S = \{a, b, c, d, e, v_1^{(1)}, v_2^{(1)}, v_3^{(1)}, v_1^{(2)}\}$. The delegate vertices are indicated with a dashed outline. Notice that the neighbors of $b$, aside from $c$, are not precisely specified in the code and are instead treated as delegate vertices.
    }
    \label{fig:example2}
    \end{figure*}
    
    Substituting the approximate numerator and denominator, we have
    \begin{align*}
        &\prob \left(\mathrm{des}\big(v^{(i)}_j\big) = d^{(i)}_j \,\text{ for all } i \in \{1,2,\dots,k\} \text{ and } j \in \{1,2,\dots,r_i\} \mid \mathcal{H}_t\right) \\
        &\qquad\sim e^{-k} \left(1+\frac{k-\sum_{i=1}^k r_i}{n-N+\sum_{i=1}^k r_i} \right)^{n-N+k-2} \prod_{i=1}^k \left( \frac{d_i}{r_i} \prod_{j=1}^{r_i} \frac{1}{d^{(i)}_j!} \right)\\
        &\qquad \to \prod_{i=1}^k \left( \frac{d_i}{r_i} \prod_{j=1}^{r_i} \frac{e^{-1}}{d^{(i)}_j!} \right)
    \end{align*}

    Next, consider the exploration process for $(\ell^S_\infty, T^S_\infty, \root)$ and condition on the subgraph up to generation $t$ to be $\mathcal{H}_t$. Note that $\mathcal{H}_t$ does not specify which vertex, in each connected component at generation $t$, has a label ending in $\tilde{0}$. For each component $i$, let the vertices at generation $t$ be denoted by $\{ v^{(i)}_{1,\infty}, v^{(i)}_{2,\infty}, \dots, v^{(i)}_{r_i,\infty} \}.$ We have
    \begin{align*}
        &\mathbb{P} \left(\mathrm{des}\big(v^{(i)}_{j,\infty}\big) = d^{(i)}_j \,\text{ for all } i \in \{1,2,\dots,k\} \text{ and } j \in \{1,2,\dots,r_i\} \mid \mathcal{H}_t\right) \allowdisplaybreaks\\
        &\qquad = \prod_{i=1}^k \mathbb{P} \left(\mathrm{des}\big(v^{(i)}_{j,\infty}\big) = d^{(i)}_j \,\text{ for all } j \in \{1,2,\dots,r_i\} \mid \mathcal{H}_t\right) \allowdisplaybreaks\\
        &\qquad = \prod_{i=1}^k \sum_{j=1}^{r_i} \frac{1}{r_i} e^{-r_i} \frac{d^{(i)}_j}{d^{(i)}_1! \, d^{(i)}_2! \ldots d^{(i)}_{r_i}!} \allowdisplaybreaks\\
        &\qquad = \prod_{i=1}^k \left( \frac{d_i}{r_i} \prod_{j=1}^{r_i} \frac{e^{-1}}{d^{(i)}_j!} \right),
    \end{align*}
    where $d_i \coloneqq \sum_{j=1}^{r_i} d^{(i)}_j$. The factor $\tfrac{1}{r_i}$ arises from choosing the $\tilde{0}$-labeled vertex uniformly among the $r_i$ vertices, while the factor $d_j^{(i)}$ accounts for selecting exactly one among its $d_j^{(i)}$ offspring to carry on the special label.  
    
    \end{proof}
    \section{Minimum Spanning Forest and Joint Local Weak Convergence}\label{app:localweakconv_tree}
    A \emph{minimum spanning forest} of an infinite graph is a natural extension of the minimum spanning tree concept to infinite graphs. 
    For a finite graph $(\ell_n, G_n)$ with $n$ vertices, an edge $e = (u,v)$ belongs to the MST precisely when, upon removing from $G_n$ all edges whose lengths are at least $\ell_n(e)$, there is no path connecting $u$ and $v$. In the infinite setting, consider a graph $(\ell_\infty, G_\infty)$ with countably 
    many vertices. We extend the finite definition by introducing a \emph{virtual vertex at infinity}, so that any path to infinity is interpreted as a path to this vertex.
    
    Formally, for an infinite weighted graph $(\ell_\infty, G_\infty)$ in which all edges have distinct weights, its minimum spanning forest $M_\infty$ is defined as the subgraph on the same vertex set that includes any edge $e=(u,v)$ satisfying:
    \begin{enumerate}
        \item $\mathcal{C}(u; G_\infty, e)$ and $\mathcal{C}(v; G_\infty, e)$ are disjoint, and 
        \item At least one of $\mathcal{C}(u; G_\infty, e)$ or $\mathcal{C}(v; G_\infty, e)$ is finite.
    \end{enumerate}
    Here, $\mathcal{C}(u; G_\infty, e)$ denotes the connected component containing $u$ in the subgraph obtained by removing all edges whose weight is at least 
    $\ell_\infty(e)$. In other words, $e \in M_\infty$ if and only if, after removing all edges with weight at least $\ell_\infty(e)$, $u$ and $v$ lie in distinct 
    components and cannot be reconnected even through the virtual vertex.

    It follows directly from the definition that the connected component of any vertex $u$ in $M_\infty$ forms a tree with infinitely many vertices. Thus, through the virtual vertex, all vertices in $M_\infty$ are connected, so $M_\infty$ is indeed a spanning forest. Moreover, under the assumption of distinct edge weights, it is straightforward to show that $M_\infty$ is unique—just as is the case for minimum spanning trees. 

    Suppose that a sequence of random finite weighted graphs $\{(\ell_n, G_n)\}_{n\ge1}$ converges in the local weak sense to a random infinite rooted weighted graph $(\ell_\infty, G_\infty, \root)$. In \cite{Aldous1991Asymptotic}, it is proved that if the limiting graph $(\ell_\infty, G_\infty, \root)$ almost surely has distinct edge weights, then the minimum spanning trees (MSTs) of the finite graphs converge in the local weak sense to the minimum spanning forest (MSF) of $(\ell_\infty, G_\infty, \root)$.

    Given the local weak convergence of a sequence of random finite weighted graphs $\{(\ell_n, G_n)\}_n$ to a random infinite rooted graph $(\ell_\infty, G_\infty, \root)$, \cite{Aldous1991Asymptotic} proved that the minimum spanning trees of the finite graphs converge locally weakly to the minimum spanning forest of $(\ell_\infty, G_\infty, \root)$, provided that all edges in the infinite graph have distinct weights with probability $1$.

    \begin{theorem}[MST Convergence Theorem~\cite{AldousSteel2001Objective}]
        Let $(\ell_\infty, G_\infty, \root)$ be a $\mathcal{G}_\circ$-valued random variable such that, with probability one, $G_\infty$ has infinitely many vertices and no two edges of $G_\infty$ have the same length. Suppose that a sequence of random finite weighted graphs $\{(\ell_n, G_n, \root)\}_n$ converges locally weakly to $(\ell_\infty, G_\infty, \root)$.
        For each $n$, let $M_n$ denote the almost surely unique MST of $G_n$, and let $M_\infty$ denote the almost surely unique MSF of $G_\infty$. Then, the sequence $\{(\ell_n, G_n, M_n, \root)\}_n$ converges jointly locally weakly to $(\ell_\infty, G_\infty, M_\infty, \root)$. 

        Further, if $N_n$ denotes the degree of a uniformly selected vertex in $M_n$ and $N$ denotes the degree of $\root$ in $M_\infty$, then
        \begin{align*}
            N_n \xrightarrow{d} N \quad \text{and} \quad \mathbb{E}[N_n] \to \mathbb{E}[N] = 2,
        \end{align*}
        and, if $S_n$ denotes the sum of the lengths of the edges incident to a uniformly selected vertex in $M_n$ and ${S}_\infty$ denotes the corresponding quantity in $M_\infty$, then 
        \begin{align*}
            S_n\xrightarrow{d} S_\infty.
        \end{align*}
    \end{theorem}
\end{document}